\documentclass{amsart}[11pt]
\usepackage{amsmath, amsfonts, amsthm, amssymb,mathtools}
\usepackage{subfigure}
\usepackage{stmaryrd}
\usepackage{verbatim}
\usepackage{hyperref}
\usepackage{color}
\usepackage{bm}
\usepackage{ulem}
\usepackage{dsfont}
\usepackage{graphicx}
\usepackage{enumerate}

\usepackage{upgreek}

\numberwithin{equation}{section}
\newtheorem{theorem}{Theorem}[section]

\newtheorem{lemma}[theorem]{Lemma}
\newtheorem{proposition}[theorem]{Proposition}

\theoremstyle{definition}
\newtheorem{definition}[theorem]{Definition}
\newtheorem{remark}[theorem]{Remark}

\newtheorem{assumption}[theorem]{Assumption}

\newcommand{\ind}{\mathds{1}} %Indicator Function
\newcommand{\I}{\mathtt{i}}

\newcommand{\RR}{\mathbb{R}}

\newcommand{\PP}{\mathbb{P}}

\newcommand{\NN}{\mathbb{N}}
\newcommand{\Kk}{\mathcal{K}}

\newcommand{\D}{\mathrm{d}}

\newcommand{\CC}{\mathbb{C}}

\newcommand{\Ee}{\mathcal{E}}
\newcommand{\Tt}{\mathcal{T}}

\newcommand{\E}{\mathrm{e}}

\newcommand{\EE}{\mathbb{E}}

\newcommand{\Nn}{\mathfrak{N}}
\newcommand{\Vv}{\mathcal{V}}
\newcommand{\Ff}{\mathcal{F}}

\newcommand{\eps}{\varepsilon}

\newcommand{\TT}{\mathbb{T}}
\newcommand{\nn}{\mathfrak{n}}
\newcommand{\ZZ}{\mathfrak{Z}}
\newcommand{\zz}{\mathfrak{z}}
\newcommand{\ssf}{\mathfrak{s}}
\newcommand{\ee}{\mathfrak{e}}
\newcommand{\Qq}{\mathcal{Q}}
\newcommand{\Qf}{\mathfrak{Q}}
\newcommand{\qf}{\mathfrak{q}}

\DeclareMathOperator*{\foo}{\text{\raisebox{-0.7ex}{\huge $\Phi$}}}
\newcommand{\Uu}{\mathcal{U}}
\newcommand{\Aa}{\mathcal{A}}

\begin{document}
\title{Optimal liquidation in a Level-I limit order book for large-tick stocks}
\date{\today}
\author{Antoine Jacquier}
\address{Department of Mathematics, Imperial College London}
\email{a.jacquier@imperial.ac.uk}
\author{Hao Liu}
\address{Department of Mathematics, Imperial College London}
\email{hao.liu112@imperial.ac.uk}

\thanks{The authors would like to thank Martin Gould and Fabrizio Lillo for useful discussions.
AJ acknowledges financial support from the EPSRC First Grant EP/M008436/1 and 
HL acknowledges financial support from SHELL.
}

\maketitle

\begin{abstract}
We propose a framework to study the optimal liquidation strategy in a limit order book for large-tick stocks, 
with the spread equal to one tick. 
All order book events (market orders, limit orders and cancellations) occur according to independent Poisson processes, with parameters depending on the most recent price move direction.
Our goal is to maximise the expected terminal wealth of an agent who needs to liquidate her positions within a fixed time horizon.
By assuming that the agent trades (through both limit and market orders) only when the price moves, 
we model her liquidation procedure as a semi-Markov decision process, 
and compute the semi-Markov kernel using Laplace method in the language of queueing theory. 
The optimal liquidation policy is then solved by dynamic programming,
and illustrated numerically. 
\end{abstract}
%%%%%%%%%%%%%%%%%%%%%%%%%%%%%%%%%%%%%%
%%%%%%%%%%%%%%%%%%%%%%%%%%%%%%%%%%%%%%
\section{Introduction}
Nowadays, most equity and derivative exchanges all over the world are at least partially using order-driven trading mechanisms:
Helsinki, Hong Kong, Shenzhen, Swiss, Tokyo, Toronto, Vancouver Stock Exchanges, 
Australian Securities Exchange, Euronext
are pure order-driven markets 
and New York, London Stock Exchanges, Nasdaq are hybrid markets~\cite{gould2013limit}.
Different from a quote-driven market, where large market makers centralise buy and sell orders and provide liquidity to other market participants through setting the bid and ask quotes,
an order-driven market is much more flexible,
which allows all market participants to send buy or sell orders specifying the price and amount they want to trade into a limit order book (LOB). 
According to the classical terminology~\cite[Section~2.2]{gould2013limit}, 
orders leading to an immediate execution upon submission based on the LOB's trade-matching algorithm are called market orders,
while orders that do not result in an immediate execution and therefore are stored in the LOB are called limit orders.
The active limit orders can either get executed by subsequent counterpart market orders based on a certain priority rule\footnote{
A priority rule regulates how limit orders stored in the LOB will get executed.
By far the most common priority rule is `{price-time}'~\cite[Section 3.4]{gould2013limit},
that is,
limit orders posted closer to the mid price will get priority and
limit orders posted at the same price follow the `first come first serve' rule.}
or be cancelled.
Therefore, a LOB can be understood as a collection of buy and sell limit orders stored at different price levels awaiting to be executed by counterpart market orders or cancellations. 
%%%%%%%%%%%%%%%%%%%%%%%%%%%%%%%%%%%%%%%

The predominance of automated order-driven markets, 
together with the significant breakthroughs in quantitative modelling and information technology in recent years, has vastly facilitated the emergence and proliferation of algorithmic trading.
Broadly speaking, algorithms serve different purposes and are classified as either proprietary or agency~\cite{hasbrouck2013low}.
Proprietary algorithms are mainly employed by high-frequency traders
aiming at making profits from the trading process itself~\cite{hagstromer2013diversity, menkveld2016economics, securities2014equity}.
Agency algorithms, on the other hand, are normally used by buy-side institutional investors to implement long-term position changes,
aiming at minimising the execution cost and market impact.
In the process of buying or selling a large parent order,
an agency algorithm is essentially decomposed into three layers~: 
\begin{enumerate}
\item [\textbf{(L$1$)}]
how to slice the parent order and schedule the child orders over the entire trading horizon;
\item [\textbf{(L$2$)}]
what is the price, type and timing to execute each child order within the scheduled horizon;
\item [\textbf{(L$3$)}]
which venue(s) should each child order be routed to.
\end{enumerate}
Almgren and Chriss~\cite{almgren2001optimal}, 
Almgren~\cite{almgren2003optimal}, 
Gatheral, Schied and Slynko~\cite{gatheral2012transient}
and 
Lorenz and Almgren~\cite{lorenz2011mean}
address the optimal execution problems by solely considering the first layer, 
in which case the direct interactions between the trader and the LOBs are abstracted away.
Some studies take the first two layers into account and
formulate optimal strategies for executing a large position in a single LOB market.
For example, Obizhaeva~and~Wang~\cite{obizhaeva2013optimal}
and Alfonsi, Fruth and~Schied~\cite{alfonsi2010optimal}
develop the optimal execution strategies entirely using market orders, 
assuming that the liquidity replenishes gradually over time after it is taken.
Bayraktar and Ludkovski~\cite{bayraktar2014liquidation},
Gu{\'e}ant, Lehalle and Fernandez-Tapia~\cite{gueant2012optimal}
design the optimal liquidation strategy that posts limit orders only, 
treating the liquidation process as a sequence of order fills and modelling it by a point process.
Cartea and Jaimungal~\cite{cartea2015optimal} seek to execute a large order employing both market and limit orders,
and solve the optimal strategies under different scenarios.
Some researchers focus on the second layer and study how to optimally execute a single child order,
for the purpose of incorporating information on the LOB market microstructure into their trading strategies,
in particular Stoikov and Waeber~\cite{stoikov2012optimal},
Donnellya and Gan~\cite{donnelly2017optimal}
and
Gonzalez and Schervish~\cite{gonzalez2017instantaneous}
for market-order-oriented, limit-order-oriented and hybrid optimal strategy, respectively. 
Finally, Cont and Kukanov~\cite{cont2017optimal} 
combine the last two layers together 
and propose a strategy that optimally distributes a child order across different order types and trading venues.

%%%%%%%%%%%%%%%%%%%%%%%%%%%%%%%%%%%%%%%%%%%%%%%%%%%%
In this paper, we formulate and solve a stylised optimal liquidation problem from the perspective of the second layer of an agency algorithm.
Specifically, we consider an agent (or her agency algorithm) who wants to sell a child order of a pre-specified (small) quantity over a fixed (short) trading window in a LOB of a large-tick stock\footnote{See~\cite[Section 4]{gould2015queue} for definition and selection criteria of large-tick stocks.},
where the price-time priority mechanism is applied.
Information available to this agent contains historical order flows and depths of the LOB at the best prices (`Level-I' data).
In particular, we are mostly interested in how different trading conditions 
(LOB state, inventory position, time to maturity) impact the agent's decisions.
In order to achieve this,
we first build up a `Level-I' LOB model describing the trading environment 
whose dynamics are driven by the general market participants' order flows and exogenous information. 
Realistic simplifying assumptions for this LOB follow those in~\cite{cont2013price, cont2010stochastic},
including unit order size, 
constant one-tick spread, 
Poisson order flows,
depletion of the best bid (resp. ask) queue moving the price one tick downward (resp. upward)
and volumes at best prices after a price move being regarded as stationary variables drawn from a joint distribution. 
We further develop this model by allowing the Poisson rates of the order flows 
and the joint distribution determining the depths at the best prices after a price move
to depend on the most recent price move direction.
Under these assumptions, the evolution of this LOB can be modelled as a Markov renewal process as in~\cite{fodra2015semi},
whose transition mechanism is intuitively described by a queueing race between the volumes at the best prices.
We then assume the agent to be risk-neutral,
trying to maximise her expected terminal wealth by selling a fixed-amount child order within a fixed (finite) time horizon in this LOB.
In order to model the price-time priority rule and capture the executions of the agent's limit orders, we assume that the agent is slow and only reacts immediately after the price moves using both limit and market orders:
at each price-change time,
the agent can choose to post a limit order at the best ask price with the least time priority 
and/or submit a market order that never consumes up the entire volumes at the best bid price.
Through combining the assumptions for the LOB and the liquidating strategy,
the agent's trading procedure is then formulated through a (stationary) semi-Markov decision process within a finite horizon~\cite{huang2011finite},
among a certain class of horizon-related Markov deterministic policies.
In general, at each price-change time, the optimal policy is a deterministic function
which tells the agent the size of the market and limit order to trade based on the current LOB state
(price move direction, volumes at the best prices), the agent's inventory position and time to maturity in order to achieve terminal wealth maximisation.
%%%%%%%%%%%%%%%%%%%%%%%%%%%%%%%%%%%%%%%%%%%%%%%%%%%%

We restrict our attention to optimal execution strategy for large-tick stocks 
mainly because the spread of large-tick stocks is almost always equal to one tick~\cite{dayri2015large}, 
so that in most cases traders cannot undercut each other by submitting limit orders inside the spread and therefore have to wait in the queue to get executed. 
This feature may largely simplify the LOB modelling. 
More importantly, market conditions and trading strategies for large-tick stocks are deemed to be different from those for small-tick stocks~\cite{o2015relative}. 
Therefore, the trading strategies for these two categories of stocks should be studied separately.

%%%%%%%%%%%%%%%%%%%%%%%%%%%%%%%%%%%%%%%%%%%%%%%%%%%%
This paper is organised as follows. 
In Section~\ref{sec:stylisedLOB}, 
we set the basic assumptions for the LOB model, illustrate the evolutional dynamics of a `Level-I' LOB and define the objective together with the admissible trading strategy set for the agent.
In Section~\ref{sec:TradingProcedure}, 
a semi-Markov decision process with a horizon-related Markov deterministic policy is introduced to model the agent's trading procedure and an optimal policy is defined. 
In Section~\ref{Section:SemiMarkovKernel}, we provide an expression for the semi-Markov kernel,
which works as the transition mechanism of the semi-Markov decision process. 
Existence of a stationary optimal policy is proved in Section~\ref{Section:OptimalStrategy}, 
and empirical studies show our numerical results in Section~\ref{sec:empirical}.
\\
%%%%%%%%%%%%%%%%%%%%%%%%%%%%%%%%%%%%%%%%%%%%%%%%%%%%
%%%%%%%%%%%%%%%%%%%%%%%%%%%%%%%%%%%%%%%%%%%%%%%%%%%%

\noindent\textbf{Notations:}\label{notations}
we shall use the following notations:
$\NN := \{0, 1, 2, \dots\}$, 
$\NN^+ := \{1, 2, \dots\}$, 
$\RR^+:=(0,\infty)$, $\RR^+_0:=\RR^+\cup\{0\}$, $\RR^-:=\RR\setminus\RR^+_0$,
and~$\CC$ represents the set of imaginary numbers.
In this paper, $T>0$ is a fixed terminal time, and we denote
$\TT := [0,T]$ and $\TT_- := \RR^-\!\cup\TT$.
For a continuous-time process~$({L}_s)_{s\geq 0}$, 
denote~$\tau_{{L}}$ its first passage time to the origin, 
and~$f_L$ (resp.~$F_L$) the density (resp. cumulative distribution function) of~$\tau_{L}$ and, as usual $\overline{F}_L := 1-F_L$.

%%%%%%%%%%%%%%%%%%%%%%%%%%%%%%%%%%%%%%%%%%%%%%%%%%%%
\section{Limit Order Book and Trading Strategy}\label{sec:stylisedLOB}
\subsection{`Level-I' Limit Order Book Model}\label{subsec:LOB basic assumption}
We consider a limit order book characterised by two resolution parameters as in~\cite[Section~2.1]{gould2013limit}: 
the tick size~$\eps>0$ represents the smallest interval (assumed constant) between price levels,
and the lot size,~$\sigma >0$, specifies the smallest amount of the asset that can be traded.
All buy and sell orders thus must arrive at a price~$k_1\eps$
and with a size~$k_2\sigma$, for some~$k_1, k_2\in\NN^+$.
Throughout this paper we shall work with the following modelling assumptions for the limit order book:
\begin{assumption}[Order book settings]\label{ass:Underlying}\ 
\begin{enumerate}[(a)]
\item \label{ass:Underlying1} 
orders from general market participants are of unit size, 
defined by~$\sigma$ actual size;
\item \label{ass:Underlying2} the spread of the limit order book is equal to the tick size~$\eps$.
\end{enumerate}
\end{assumption}
The LOB model is formulated based on a `Level-I' data, that is, the order flows and depths at the best bid and ask prices. 
As illustrated in~\cite[Section~2.1]{cont2013price},
this reduced-form modelling approach is motivated by empirical findings 
showing (a) that large amounts of order flows occur at the best price levels for large-tick stocks~\cite{gareche2013fokker},
(b) that the imbalance between the order flows at the best prices is shown to be a good predictor of the order book dynamics~\cite{cartea2015enhancing, cont2014price},
and (c) that data at the best prices are more obtainable than the `Level-II' market data.
In the following, we impose the assumptions for the evolution of the LOB:
\begin{assumption}[Evolution of the limit order book]\label{ass:LOBEvolution}\leavevmode
\begin{enumerate}[(a)]
\item \label{ass:LOBEvolution1} 
whenever orders at the best bid (resp. ask) price are depleted, 
both the best bid and ask prices decrease (resp. increase) by one tick;
\item \label{ass:LOBEvolution3} 
immediately after each price increase (resp. decrease), volumes at the best bid and ask prices are treated as random variables with joint distribution
$f_{+1}$  (resp. $f_{-1}): (\NN^+)^2\to[0, 1]$;
for any~$x_1, x_2\in\NN^+$, $f_{+1}(x_1 ,x_2)$ (resp. $f_{-1}(x_1, x_2)$) 
represents the probability that the best bid and ask queue contain~$x_1$ and~$x_2$ unit limit orders
(of actual size~$x_1\sigma$ and~$x_2\sigma$), 
right after a price increase (resp. decrease).
\end{enumerate}
\end{assumption}
\begin{remark}
Assumption~\ref{ass:LOBEvolution} presumes that the limit order book contains no empty level near mid price
so that price changes are restricted to one tick,
and that price changes are entirely due to exogenous information,
in which case market participants swiftly readjust their order flows at the new best prices, 
as if a new state of the limit order book is drawn from its invariant distribution~\cite{huang2015simulating}.
In other words, we rule out the possibility that depletion of the best bid (resp. ask) queue is followed by the insertion of a buy (resp. sell) limit order inside the spread,
keeping the best bid and ask prices unchanged.
\end{remark}

Modelling order flows from general market participants is based on the `zero-intelligence' approach~\cite{cont2013price, cont2010stochastic, smith2003statistical}.
\begin{assumption}[Poisson order flows]\label{ass:PoissonOrderFlow}\leavevmode
All order book events (market orders, limit orders and cancellations) from general market participants occur according to independent Poisson processes, 
with parameters depending on the most recent price move direction. 
To be more specific,
taking order flows at the best ask price for example, during any period between a price increase (resp. decrease) and the next price change, the following mutually independent events happen:
\begin{enumerate}[(a)]
\item \label{ass:PoissonOrderFlow1}
buy market orders arrive at independent, exponential times with rate $\mu^a_{+1}>0$ (resp. $\mu^a_{-1}>0$);
\item \label{ass:PoissonOrderFlow2}
sell limit orders arrive at independent, exponential times with rate $\kappa^a_{+1}>0$ (resp. $\kappa^a_{-1}>0$);
\item \label{ass:PoissonOrderFlow3}
cancellations of limit orders occur at independent, exponential times with rate $\theta^a_{+1}>0$ (resp. $\theta^a_{-1}>0$) multiplied by the amount (in unit size) of the outstanding sell limit orders. 
\end{enumerate}
We assume an analogous framework at the best bid price, with parameters
$\mu^b_{+1}, \mu^b_{-1}, \kappa^b_{+1}, \kappa^b_{-1}, \theta^b_{+1}, \theta^b_{-1} >0$.
\end{assumption}
\begin{remark}\ 
\begin{itemize}
\item Although the `zero-intelligence' model is not exactly compatible with empirical observations~\cite{zhao2010model}, 
it still retains the major statistical features of limit order books while remaining computationally manageable.
With the `zero-intelligence' hypothesis,
the agent can easily characterise the dynamical properties of the limit order book from historical data without assuming behavioural assumptions for other market participants or resorting to auxiliary assumptions to quantify unobservable parameters.
\item Assumption~\ref{ass:PoissonOrderFlow}(\ref{ass:PoissonOrderFlow3}) means that 
if there are~$v$ limit orders at the best ask (resp. bid) price, each of which can be cancelled at an exponential time with rate~$\theta^a$ (resp.~$\theta^b$) independently,
and the overall cancellation rate is then~$\theta^a v$ (resp.~$\theta^b v$).
\end{itemize}
\end{remark}
%%%%%%%%%%%%%%%%%%%%%%%%%%%%%%%%%%%%%%%%%%%%%%%%%%%%%%%%%%%%%%%%%%%%%%%%%%%%%%%%%%%%%%

\subsection{Objective and admissible trading strategies}\label{subsec:strategyassum}
In the limit order book model introduced in Section~\ref{subsec:LOB basic assumption},
we assume that the agent is risk-neutral and her goal is to maximise the expected wealth obtained through selling the child order of~$\chi\in\NN^+$ unit size ($\chi\sigma$ actual size) within the finite horizon~$\TT$.
The following assumption describes the set of admissible trading strategies:
\begin{assumption}[Admissible trading strategies]\label{ass:AdmissibleTradingS}\ 
\begin{enumerate}[(a)]
\item \label{itm:Strategyfirst}
the agent can only trade immediately after a price change;
let~$\tau_n$ denote her~$n$-th decision epoch, 
namely the time of the~$n$-th price change; 
$\tau_0 = 0$ and the last decision epoch before or at maturity is~$\tau_{\nn}$, 
where~$\nn:=\sup\{n\in\NN: \tau_n\leq T\}$;
\item \label{itm:Strategyfirst'}
at maturity~$T$,
the agent is required to sell all the unexecuted stocks through a market order;
\item \label{itm:Strategysecond}
at each decision epoch~$\tau_n$, the agent observes the bid and ask queues,
with volumes of~$v^b$ and $v^a$ unit size;
she can then post a sell limit order of~$l$ unit size at the best ask price
and submit a sell market order of~$m$ unit size at the best bid price;
we assume that the best bid queue is never depleted by the agent,
and that the agent is slow, meaning that her limit order (of~$l$ unit size) has less time priority upon submission than the limit orders from other market participants (of~$v^a$ unit size);
\item \label{itm:Strategythird}
the agent follows a `no cancellation' rule:
she will not cancel her limit order unless the price goes down;
\item \label{itm:Strategyfourth}
short selling is not allowed.
\end{enumerate}
\end{assumption}
Restricting the agent's trading actions at price changes 
(Assumption~\ref{ass:AdmissibleTradingS}\eqref{itm:Strategyfirst}) might sound relatively strong,
but is necessary to capture the time-priority rule and the executions of the agent's limit orders. 
We shall study later in Section~\ref{sec:OptimalStrategy} how to define an optimal policy
maximising the expected wealth at maturity~$T$. 
%%%%%%%%%%%%%%%%%%%%%%%%%%%%%%%%%%%%%%%%%%%%%%%%%%%%%%%%%%%%%%%%%%%%%%%%%%%%%%%%%%%%%%
\section{Trading procedure modelled by semi-Markov decision processes}\label{sec:TradingProcedure}
A semi-Markov decision model~\cite[Chapter~7]{tijms2003first} 
is a dynamic system whose states are observed at random epochs,
each of when an action is taken and a payoff incurs
(either as a lump sum at that epoch or at a rate continuously until the next epoch)
as a result of the action.
It satisfies the following two Markovian properties:
\begin{enumerate}
\item[\textbf{(M$1$)}] given the current state and the action at a given epoch,
 the time until the next epoch and the next state only depend on the current state and action;
\item[\textbf{(M$2$)}] the payoff incurred at any epoch depends only on the state and the action at that epoch.
\end{enumerate}
The semi-Markov decision model well describes the agent's liquidation problem within our stylised limit order book:
the limit order book with the agent's participation is a dynamic system, 
and the agent's selling action at each decision epoch may lead to a payoff.
Indeed, Assumption~\ref{ass:AdmissibleTradingS}\eqref{itm:Strategyfirst} enables us to track the state of this system merely at the decision epochs, 
and Assumptions~\ref{ass:LOBEvolution}, \ref{ass:PoissonOrderFlow} and~\ref{ass:AdmissibleTradingS}\eqref{itm:Strategysecond}
ensure that the transition mechanism of the system is stationary and satisfies~\textbf{(M$1$)}-\textbf{(M$2$)}.
Moreover, according to Assumption~\ref{assumption:limitpayoffend},
each payoff from the agent's matched limit order is allocated to the nearest incoming decision epoch
in order to make the payoff as a lump sum. 
In Section~\ref{sec:IntoSMDP}, we define a (stationary) semi-Markov decision model with lump-sum payoffs for the agent's liquidation procedure.
In Section~\ref{sec:evolutionSMDP}, we define a horizon-related Markov deterministic policy and illustrate the evolution of the semi-Markov decision process.
In Section~\ref{sec:vfOptimalStrategy}, we give the definition of the expected reward function, the value function and the optimal policy for the agent's liquidation problem.
%%%%%%%%%%%%%%%%%%%%%%%%%%%%%%%%%%%%%%%%%%%

\subsection{Semi-Markov decision model}\label{sec:IntoSMDP}
The semi-Markov decision model with lump-sum payoffs and the finite-horizon constraint is defined as a six-tuple
$\left\{\Ee, (\Aa(e))_{e\in \Ee}, Q(\cdot, \cdot\lvert \cdot), P(\cdot\lvert\cdot), r(\cdot, \cdot), w(\cdot, \cdot)\right\}$,
where each element is defined below.
%%%%%%%%%%%%%%%%%%%%%%%%%%%%%%%
\subsubsection{State space}
Fix~$N\in\NN^+$ large enough.
The state space~$\Ee:= \{-1, +1\}\times\{1, \dots, N\}^3\times\{0, \dots, N\}^2$
is the set of all pre-decision conditions of the system
(i.e. the limit order book with the agent's participation)
observed at each decision epoch.
Specifically, the system being in state~$e := (j, v^b, v^a, p, z, y) \in \Ee$ means that: 
\begin{itemize}
\item the ask/bid price change is equal to~$j$ tick;
\item the best bid (resp. ask) queue contains~$v^b$ (resp. $v^a$) unit orders;
\item the ask price\footnote{The stylised limit order book model doesn't implement a positive restriction on the stock price.
But we assume that the stock price is far above zero at inception
and the liquidation horizon~$\TT$ is short,
so that the stock price will never become negative.} is equal to~$p\eps$;
\item the executed part of the limit order posted by the agent at the previous decision epoch is of~$z$ unit size;
\item the agent's remaining inventory position is of~$y$ unit size.  
\end{itemize}

%%%%%%%%%%%%%%%%%%%%%%%%%%%%%%%%%%%%
\subsubsection{Action space}\label{sec:actionspace}
The action space ${\Aa} := \{0, \dots, \overline{m}\}\times\{0, \dots, \overline{l}\}$, 
with~$\overline{m}, \overline{l}\in\NN^+$,
represents the set of trading strategies, that is,
the amount (in unit size) of the market and limit order that the agent chooses to submit and post at the best bid and ask price respectively.
The constant~$\overline{m}$ (resp. $\overline{l}$)
represents the maximum amount (in unit size) of a single market (resp. limit) order that the agent is allowed to trade. 
From Assumption~\ref{ass:AdmissibleTradingS}\eqref{itm:Strategysecond}\eqref{itm:Strategyfourth},
the agent's admissible action space in state~$e \in \Ee$ is defined by
\begin{equation}\label{def:actionspace}
{\Aa}(e) := \left\{(m, l)\in {\Aa}: m < v^b, m + l \leq y\right\},
\end{equation}
so that the agent will never consume up the entire best bid queue nor short sell. 
The set of all feasible state-action pairs is denoted by~$\Kk:=\{(e, \alpha) | e\in \Ee, \alpha\in\Aa(e)\}$.

%%%%%%%%%%%%%%%%%%%%%%%%%%%%%%%%%%%%%%%%
\subsubsection{Semi-Markov kernel}\label{subsubsec:semi-Markov kernel}
Before introducing our next concept, recall the following definition.
\begin{definition}[sub-/semi-Markov kernel]\label{def:kernel}
Let~$(\Omega_1, \mathcal{F}_1)$ and~($\Omega_2, \mathcal{F}_2)$ be real measurable spaces.
A map~$p(\cdot | \cdot): \mathcal{F}_2\times\Omega_1\to[0 ,1]$
is called a sub-Markov kernel on~$\Omega_2$ given~$\Omega_1$ if:
\begin{itemize}
\item for any~$\omega_1\in\Omega_1$, 
$p(\cdot | \omega_1)$ is a measure on~$(\Omega_2, \mathcal{F}_2)$ with $p(\Omega_2 | \omega_1) \leq 1$;
\item for any~$F_2\in\mathcal{F}_2$, $p(F_2\lvert\cdot)$ is a Borel measurable function.
\end{itemize}
In particular, if~$p(\Omega_2\lvert\omega_1) = 1$ for all~$\omega_1\in\Omega_1$,
then $p(\cdot\lvert\cdot)$ is a Markov kernel on~$\Omega_2$ given~$\Omega_1$.
Furthermore, a map~$q(\cdot, \cdot | \cdot): \RR^+_0\times\mathcal{F}_2\times\Omega_1\to[0, 1]$ 
is a semi-Markov kernel on~$\RR^+_0\times\Omega_2$ given~$\Omega_1$ if:
\begin{itemize}
\item for~$(F_2, \omega_1)\in\mathcal{F}_2\times\Omega_1$,
$q(\cdot, F_2\lvert \omega_1)$ is non-decreasing, right-continuous and~$q(0, F_2\lvert \omega_1) = 0$;
\item for~$t\geq0$, $q(t, \cdot | \cdot)$
is a sub-Markov kernel on~$\Omega_2$ given~$\Omega_1$;
\item the limit~$\displaystyle\lim_{t\uparrow \infty}q(t, \cdot | \cdot)$
is a Markov kernel on~$\Omega_2$ given~$\Omega_1$.
\end{itemize}
\end{definition}
In our model, let~$Q(\cdot, \cdot\lvert \cdot)$  be a semi-Markov kernel on $\RR^+_0\times \Ee$ given~$\Kk$,
determining the (stationary) transition mechanism of the semi-Markov decision process:
for any~$t\geq 0$ and~$\tilde{e}\in \Ee$, 
given the state-action pair~$(e, \alpha)\in \Kk$ at some decision epoch,
the quantity\footnote{By abuse of language,
we write~$Q(t, \{\tilde{e}\}\lvert({e}, {\alpha}))$ as~$Q(t, \tilde{e}\lvert({e}, {\alpha}))$.}
$Q(t, \tilde{e}\lvert ({e}, {\alpha}))$ represents 
the (joint) probability that the time until the next decision epoch is less than or equal to~$t$ 
and the next system state is~$\tilde{e}$. 
Detailed computations are given in Section~\ref{Section:SemiMarkovKernel}.

%%%%%%%%%%%%%%%%%%%%%%%%%%%%%%%%%%%%%%%%%%%%%%%%%
\subsubsection{Terminal kernel}
The terminal kernel~$P(\cdot\lvert\cdot)$ 
is a sub-Markov kernel on~$\NN$ given~$\Kk\times\TT_-$,
and describes the execution dynamics between the last decision epoch and the maturity:
for any~$\zz\in\NN$, given the state-action pair~$(e, \alpha)\in \Kk$ and the time to maturity~$\lambda\in\TT_-$ at some decision epoch\footnote{A decision epoch with time to maturity~$\lambda < 0$ means that it happens a period of time~$\lvert\lambda\lvert$ after the maturity.},
the quantity\footnote{By abuse of language,
we write~$P\left(\{\zz\}\lvert\left((e, \alpha), \lambda\right)\right)$ as~$P(\zz\lvert(e, \alpha), \lambda)$.} 
$P(\zz\lvert(e, \alpha), \lambda)$
represents the (joint) probability that the time until the next decision epoch is strictly larger than~$\lambda$ and the executed part of the limit order up to the maturity is of~$\zz$ unit size.
Detailed computations are given in Section~\ref{Section:SemiMarkovKernel}.
\begin{remark}\label{rmk: tk}
According to our modelling framework, the terminal kernel satisfies the following properties:
\begin{itemize}
\item 
$P(0\lvert (e, \alpha), \lambda) = 1$ when~$\lambda\leq 0$;
\item 
$\sum_{\zz \geq 0}P(\zz\lvert (e, \alpha), \lambda) = 1 - Q(\lambda, \Ee\lvert (e, \alpha))$ when~$\lambda>0$;
\item $P(\zz\lvert(e, \alpha), \lambda) = 0$  when~$\zz > l$;
\end{itemize}
for any~$(e, \alpha)\in \Kk$.
\end{remark}

%%%%%%%%%%%%%%%%%%%%%%%%%%%%%%%%%%%%%%%%%%%%%%%%%
\subsubsection{Periodical reward function}
The {periodical reward function}~$r: \Kk\to\RR^+_0$
is defined as
\begin{equation}\label{eq:p_PFun}
r(e, \alpha) :=\rho\left[m\left(p - 1\right) + z\left(p - j\right)\right],
\qquad\text{for all } (e, \alpha)\in \Kk,
\qquad\qquad\text{where }\rho := \eps \sigma,
\end{equation}
and represents the lump-sum payoff associated with a decision epoch
given the state-action pair~$(e, \alpha)$.
Specifically, 
the definition~\eqref{eq:p_PFun} is given based on 
the following assumption that assigns the payoff from the matched part of the agent's limit order to the nearest incoming decision epoch.
\begin{assumption}[Periodic reward function]\label{assumption:limitpayoffend}
For $n \in \NN^+$, the payoff from the matched limit order within the interval~$[\tau_{n-1}, \tau_{n})$ is allocated at $\tau_{n}$. 
\end{assumption}
Assuming that the system is in state~$e\in \Ee$ and the agent takes action~$\alpha\in\Aa(e)$ at some decision epoch.
She then earns an immediate payoff worth~$m(p - 1)\rho$ from submitting the market order of~$m$ unit size at the best bid price~$(p - 1)\eps$.
On top of that, the matched limit order of $z$ unit size at the previous best ask price~$(p- j)\eps$ entails a payoff worth~$z(p - j)\rho$,
which is allocated at the current decision epoch according to Assumption~\ref{assumption:limitpayoffend}.
%%%%%%%%%%%%%%%%%%%%%%%%%%%%%%%%%%%%%%%%%

\subsubsection{Terminal reward function}
The terminal reward function $w:\Kk\times\NN\to\RR^+_0$ is defined as
\begin{equation}\label{eq:t_PFun}
w(e, \alpha, \zz) := 
\rho\left[\left(p - 1\right)\left(y - m\right) + \zz\right] - g\left(y - m - \zz\right),
\qquad\text{for all }(e, \alpha)\in \Kk \text{ and } \zz\in\NN,
\end{equation}
where the market impact function $g: \NN\to \RR^+_0$ is of the form
\begin{equation}\label{eq:mifg}
g(x) :=
\rho\frac{x}{\overline{v}},
\end{equation}
for a constant $\overline{v}\in\NN^+$.
For any~$(e, \alpha)\in \Kk$ and~$\zz\in\NN$,
the quantity~$w(e, \alpha, \zz)$
represents the lump-sum payoff associated with the maturity~$T$,
given the state-action pair~$(e, \alpha)$ at the last decision epoch,
and the matched part of the agent's limit order between the last decision epoch and the maturity being of~$\zz$ unit size.
Particularly,
the identity~\eqref{eq:t_PFun} is given based on the following assumption:
\begin{assumption}[Terminal reward function]\label{ass:TRF}\leavevmode
\begin{enumerate}[(a)]
\item \label{ass:TRFa}the payoff from the matched limit order obtained within the interval~$[\tau_{\nn},T)$ is allocated at~$T$;
\item \label{ass:TRFb}
when depicting the market impact brought by the market order at maturity, 
we assume that the impact is linear with~$\overline{v}$ representing the average depth (in unit size) on the bid side of the limit order book;
\item \label{ass:TRFc}
the unexecuted shares at maturity cannot sweep all the liquidity on the bid side of the limit order book, so that the terminal reward function is~$\RR^+_0$-valued. 
\end{enumerate}
\end{assumption}
Assumption~\ref{ass:TRF}\eqref{ass:TRFb} yields the market impact function~$g(\cdot)$ in~\eqref{eq:mifg}.
Furthermore, based on Assumption~\ref{ass:TRF}\eqref{ass:TRFa}\eqref{ass:TRFb},
the terminal reward~$w(e, \alpha, \zz)$
consists of the payoff from the matched limit order (of amount~$\rho p \zz$) 
and the market order at maturity (of amount~$\rho(p-1)(y-m-\zz)$), 
deducted by the corresponding market impact (of amount~$g(y -m - \zz)$). 
%%%%%%%%%%%%%%%%%%%%%%%%%%%%%%%%%%%%%%%%%%%%%%%

\subsection{Dynamics of the finite-horizon semi-Markov decision process}\label{sec:evolutionSMDP}
Assume that the agent applies a horizon-related Markov deterministic policy defined below,
specifying a decision rule for her action at each epoch based on the current state and time to maturity.
\begin{definition}\label{def:hrpolicy}
A decision rule is a measurable function
$$
\phi: \Ee\times\TT_-\ni(e,\lambda) \mapsto \alpha\in \Aa(e),
$$
such that~$\phi(e, \lambda) = (0, 0)$ for any~$(e, \lambda)\in \Ee\times\RR^-$.
Let~$\Phi$ represent the set of decision rules.
A horizon-related Markov deterministic policy is a sequence of decision rules
$$
\pi := \{\phi_0, \phi_1, \phi_2, \dots\},
$$
with $\phi_n\in\Phi$ for any~$n\in\NN$.
We denote by~$\Pi$ the set of horizon-related Markov deterministic policies.
A policy~$\pi\in\Pi$ is said to be stationary if there exists~$\phi\in\Phi$ such that~$\phi_n = \phi$ for any~$n\in\NN$
and we write~$\pi = \{\phi, \phi, \dots\} := \pi^\phi$.
We denote~$\Pi^S$ the set of stationary horizon-related Markov deterministic policies.
\end{definition}
\begin{remark}\label{rem:StopTrading}
At the~$n$-th decision epoch with system state~$e_n$
and time to maturity $\lambda_n := T - \tau_n$,
an action~$a_n = \phi_n(e_n, \lambda_n)$ is given by the decision rule~$\phi_n$
when the policy~$\pi\in\Pi$ is applied.
In particular,
the agent stops trading at any decision epoch~$\tau_n$ with $n>\nn$ (namely~$\lambda_n <0$) as~$\alpha_n = (0, 0)$ by Definition~\ref{def:hrpolicy},
fulfilling Assumption~\ref{ass:AdmissibleTradingS}\eqref{itm:Strategyfirst'}.
\end{remark}
Table~\ref{tab:example of SMDP} summarises the evolution of the semi-Markov decision model when implementing a policy~$\pi\in\Pi$.
Suppose that the system is in state~$e_0$ at inception~$\tau_0$,
and the agent has a planned trading horizon~$\lambda_0$.
According to the policy~$\pi$, she chooses the action~$\alpha_0 = \phi_0(e_0, \lambda_0)$.
It then takes a period of time~$t_1$ to reach the next decision epoch~$\tau_1 = \tau_0+t_1$,
 at which point the system state changes to~$e_1$ and the time to maturity for the agent becomes~$\lambda_1 = \lambda_0-t_1$.
She then chooses the action~$\alpha_1 = \phi_1(e_1, \lambda_1)$, and so on.
At the~$n$-th decision epoch, a periodic payoff of amount~$r(e_n, \alpha_n)$ incurs.
At maturity~$T$, a terminal payoff~$w(e_{\nn}, \alpha_{\nn},\zz)$ is obtained. 
In particular, the agent takes no action after~$T$ according to Remark~\ref{rem:StopTrading},
and correspondingly no payoff is paid.

%%%%%%%%%%%%%%%%%%%%%%%%%%%%

\begin{table}[!htp]
\centering
\resizebox{\textwidth}{!}{
\begin{tabular}{|l|l|l|l|l|l|}
\hline
Index & Time &State &Time to Maturity &Action & Payoff\\
\hline\hline
Initial & $\tau_0$ & $e_0$  & $\lambda_0\geq0$ & $\alpha_0 = \phi_0(e_0, \lambda_0)$ & $r(e_0, \alpha_0)$\\
1\textsuperscript{st} & $\tau_1 = \tau_0 + t_1$ & $e_1$ & $\lambda_1 = \lambda_0 - t_1\geq0$ & $\alpha_1 = \phi_1(e_1, \lambda_1)$ & $r(e_1, \alpha_1)$\\
2\textsuperscript{nd} & $\tau_2 = \tau_1 + t_2$ & $e_2$ & $\lambda_2= \lambda_1 - t_2\geq0$ & $\alpha_2= \phi_2(e_2, \lambda_2)$ & $r(e_2, \alpha_2)$\\
$\vdots$ & $\vdots$ & $\vdots$ & $\vdots$ & $\vdots$ & $\vdots$\\
$(\nn-1)$-th & $\tau_{\nn-1} = \tau_{\nn-2} + t_{\nn-1}$ & $e_{\nn-1}$ & $\lambda_{\nn-1} = \lambda_{\nn-2} - t_{\nn-1}\geq0$ & $\alpha_{\nn-1}= \phi_{\nn-1}(e_{\nn-1}, \lambda_{\nn-1})$ & $r(e_{\nn-1}, \alpha_{\nn-1})$\\
$\nn$-th & $\tau_\nn = \tau_{\nn-1} + t_{\nn}$ & $e_{\nn}$ & $\lambda_{\nn} = \lambda_{\nn-1} - t_\nn\geq0$ & $\alpha_{\nn}= \phi_\nn(e_\nn, \lambda_\nn)$ & $r(e_{\nn}, \alpha_{\nn})$\\
\hline
 & Terminal $T$ & & & & $w(e_{\nn}, \alpha_{\nn}, \zz)$\\
\hline
$(\nn+1)$-th & $\tau_{\nn+1} = \tau_{\nn} + t_{\nn+1}$ & $e_{\nn+1}$ & $\lambda_{\nn+1} = \lambda_\nn - t_{\nn+1}<0$ & $\alpha_{\nn+1} = (0, 0)$ & $0$\\
$\vdots$ & $\vdots$ & $\vdots$ & $\vdots$ & $\vdots$ & $\vdots$\\
\hline
\end{tabular}
}
\caption{Evolution of the semi-Markov decision process under policy~$\pi\in\Pi$}
\label{tab:example of SMDP}
\end{table}
%%%%%%%%%%%%%%%%%%%%%%%%%%%%%%%%%

In the following, we construct the semi-Markov decision process in a probability space
based on the Ionescu Tulcea's Theorem.

%%%%%%%%%%%%%%%%%%%%%%%%%%%%%%%%%
\begin{definition}\label{def: measurablespace}
Let~$(\Omega, \mathcal{F})$ be a measurable space consisting of the sample space~$\Omega$,
defined by
$$
\Omega := \Big\{\nn\in\NN, \zz\in\NN, \left(\{t_n, e_n, \lambda_n, \alpha_n\} 
\in\RR^+_0\times \Ee\times\TT_- \times\Aa(e_n)\right)_{n \in \NN}
\Big\},
$$
and the corresponding Borel $\sigma$-algebra $\mathcal{F}$.
Define the random variables~$\Nn$, $\ZZ$, $X_n$, $E_n$, $\Lambda_n$, $A_n$ on~$(\Omega, \mathcal{F})$ as:
\begin{equation*}
\begin{array}{rlrl}
\Nn(\omega) & = \nn, \qquad  & \ZZ(\omega) & = \zz,\\
X_n(\omega) & = t_n, \qquad & E_n(\omega) & := \left(J_n, V^b_n, V^a_n, P_n, Z_n, Y_n\right)(\omega) = e_n,\\
\Lambda_n(\omega) & = \lambda_n, & \qquad \qquad A_n(\omega) & :=\left(M_n, L_n\right)(\omega) = \alpha_n,
\end{array}
\end{equation*}
for any~$\omega\in\Omega$ and~$n\in\NN$,
where
\begin{itemize}
\item $X_n$ is the time between the~$(n-1)$-th and the~$n$-th decision epoch
($X_0 = 0$ almost surely);
\item $E_n, \Lambda_n, A_n$ represent the system state, time to maturity and agent's action at the~$n$-th decision epoch;
\item $\Nn$ is the index of the last decision epoch;
\item $\ZZ$ is the amount (in unit size) of the agent's limit order executed between the~$\Nn$-th decision epoch and the maturity.
\end{itemize}
\end{definition}

%%%%%%%%%%%%%%%%%%

\begin{remark}\label{rmk:evolution_E}
Based on this modelling framework, the following properties hold almost surely for~$n\in\NN$
\begin{itemize}
\item $\Lambda_{n+1} = \Lambda_n - X_{n+1}$: evolution of the time to maturity;
\item $P_{n+1} = P_n + J_{n+1}$: evolution of the ask price (in tick size);
\item $Y_{n+1} = Y_n - M_n - Z_{n+1}$: evolution of the inventory position (in unit size);
\item $Z_{n+1} \leq L_{n}$:
the amount of the matched limit order cannot exceed that of the limit order posted by the agent in each queueing race;
\item $\Nn = \sup\{n\in\NN: \Lambda_n \geq 0\}$: index of the last decision epoch;
\item $\ZZ\leq Z_{\Nn+1}$: 
the amount of the matched limit order between the last decision epoch and the maturity cannot exceed that of limit order executed when there is no finite-horizon restriction.
\end{itemize}
\end{remark}
%%%%%%%%%%%%%%%%%%
\begin{theorem}\label{thm:Tulcea}
[Tulcea's Theorem~\cite[Section 2.7.2]{ash2014real}]
For any~$(e, \lambda)\in \Ee\times\TT$ and~$\pi\in\Pi$,
there exists a unique probability measure~$\PP^{\pi}_{(e, \lambda)}$ on~$(\Omega, \mathcal{F})$ 
such that,
for any~$t\geq0$, $\tilde{e}\in \Ee$, $\alpha\in \Aa$, $\zz\in\NN$ and~$n\in\NN$,
\begin{equation*}
\begin{array}{rll}
\PP^\pi_{(e, \lambda)}(X_0 = 0, E_0 = e, \Lambda_0 = \lambda)&= 1,\\
\PP^\pi_{(e, \lambda)}(A_n = \alpha\lvert H_n = h_n)
&= \ind_{\{\phi_n(e_n, \lambda_n) = \alpha\}},
\\
\PP^\pi_{(e, \lambda)}(X_{n+1}\leq t, E_{n+1} = \tilde{e}\lvert H_n = h_n, A_n = \alpha_n)&= 
Q(t, \tilde{e}\lvert (e_n, \alpha_n)),\\
\PP^\pi_{(e, \lambda)}(X_{n+1} > \lambda_n, \ZZ = \zz\lvert H_n = h_n, A_n = \alpha_n)&= P(\zz\lvert(e_n, \alpha_n), \lambda_n),
\end{array}
\end{equation*}
where
$$
H_n := \left\{ \begin{array}{ll}
(\{X_0, E_0, \Lambda_0\}), & \text{if }n = 0,\\
\big(\{X_{i}, E_{i}, \Lambda_{i}, A_{i}\}_{i=0, \ldots, n-1}, 
\{X_n, E_n, \Lambda_n\}\big), 
& \text{if }n \in \NN^+,
\end{array}
\right.
$$
is the sequence of random variables describing the history up to the~$n$-th decision epoch
(realisations of the random variables (or sequences of random variables) are denoted by the corresponding lower case letters).
\end{theorem}

%%%%%%%%%%%%%%%%%%%%%%%%%%%%%%%%%%%%%%%%%%%%%%
\subsection{Value function and optimal policy}\label{sec:OptimalStrategy}\label{sec:vfOptimalStrategy}
Consider an agent with objective and trading strategies as described in Section~\ref{subsec:strategyassum}, introduce the following definition.
\begin{definition}\label{Def:expectedReward}
Define the {finite-horizon expected reward function} under a policy~$\pi\in \Pi$ by
\begin{equation}\label{eq:originValueFunction}
V^{\pi}(e, \lambda)
:= \EE^{\pi}_{(e, \lambda)}
\left(\displaystyle\sum_{n = 0}^{\Nn}r(E_n, A_n)
 + w(E_\Nn, A_\Nn, \ZZ)\right),
 \qquad\text{for any~$(e, \lambda)\in \Ee\times\TT$},
\end{equation}
as well as the value function
\begin{equation}\label{eq:VF}
V^*(e, \lambda) := \sup\left\{V^\pi(e, \lambda), \pi\in \Pi\right\}.
\end{equation}
A policy $\pi^*\in \Pi$ is called $\TT$-optimal if the equality
\begin{equation}\label{eq:optimalpolicy}
V^{\pi^*} (e, \lambda) = V^*(e, \lambda)
\end{equation}
holds for all $(e, \lambda)\in \Ee\times \TT$.
\end{definition}
\begin{remark}\label{rmk:vfchange}
For any $(e,\lambda) \in \Ee\times\TT$, we can rewrite the quantity~$V^\pi(e,\lambda)$ 
in~\eqref{eq:originValueFunction} as
\begin{align*}
V^\pi(e, \lambda) & =
 \EE^{\pi}_{(e, \lambda)}\left(\displaystyle\sum_{n = 0}^{\infty}
\left(r(E_n, A_n)\ind_{\{\Nn \geq n\}}\right)
 + w(E_n, A_n, \ZZ)\ind_{\{\Nn = n\}}\right)\\
 & = \EE^{\pi}_{(e, \lambda)}\left(\displaystyle\sum_{n = 0}^{\infty}
\left(r(E_n, A_n)\ind_{\{\Lambda_n\geq 0\}}
 + \displaystyle w(E_n, A_n, \ZZ)\ind_{\{0\leq\Lambda_n < X_{n+1}\}}\right)\right)\\
 & =
 \displaystyle\sum_{n=0}^\infty \EE_{(e, \lambda)}^\pi\Big(
 r(E_n, A_n)\ind_{\{\Lambda_n\geq 0\}}
 + w(E_n, A_n, \ZZ)\ind_{\{0\leq\Lambda_n < X_{n+1}\}}
\Big),
\end{align*}
where the second equality follows by writing
\begin{align*}
\{\Nn\geq n\} &= \{\Lambda_0 \geq0, \dots, \Lambda_n\geq 0\} = \{\Lambda_n \geq 0\},\\
\{\Nn = n\} & = \{\Lambda_0 \geq0, \dots, \Lambda_n\geq 0, \Lambda_{n+1}<0\} = \{\Lambda_n\geq 0, \Lambda_{n+1}<0\}
 = \{0\leq \Lambda_{n} < X_{n+1}\},
\end{align*}
since the sequence~$\{\Lambda_n\}_{n\in\NN}$ is non-increasing,
and the third equality is due to the non-negativity of the periodic/terminal reward function and the monotone convergence theorem.
\end{remark}

%%%%%%%%%%%%%%%%%%%%%%%%%%%%%%%%%%%%%%%%%%%%%%%%%%%%
\section{Semi-Markov kernel}\label{Section:SemiMarkovKernel}
We now provide the expressions for the semi-Markov kernel~$Q(\cdot, \cdot\lvert \cdot)$ 
and the terminal kernel~$P(\cdot\lvert \cdot)$
defined in Section~\ref{sec:IntoSMDP} using the language of queueing theory.
We first (Section~\ref{sec:closedformSMK}) model the dynamics of the best queues 
with the agent's participation as generalised birth-death processes, 
and derive the closed-form expressions for the semi-Markov kernel 
and the terminal kernel in all possible scenarios
in terms of the distributions of the first-passage time of the generalised birth-death processes to zero.
We then (Section~\ref{sec:queuedistribution}) compute these distributions by using Laplace method. 
%%%%%%%%%%%%%%%%%%%%%%%%%%%%%%%%%%%%%%%%%%%%%%%%%%%%
\subsection{Closed-form expressions}\label{sec:closedformSMK}
For notational convenience, we shall fix an element~$(\ee, \uplambda)$ in~$\Ee\times\TT$ 
together with a deterministic stationary policy~$\pi\in\Pi$
and denote~$\PP^{\pi}_{(\ee, \uplambda)}$ by~$\PP$
throughout this section.
%%%%%%%%%%%%%%%%%%%%%%%%%%%%%%%%%%%%%%%%%%%%%%
\subsubsection{Semi-Markov kernel}
According to Theorem~\ref{thm:Tulcea} and the Markovian property~\textbf{(M$1$)},
we can express the semi-Markov kernel as a (stationary) distribution 
of the duration and outcome of a queueing race given its initial condition and the agent's action:
\begin{equation}\label{eq:semiMarkovkernel}
Q(t, \tilde{e}\lvert ({e}, {\alpha})) = 
\PP(X_{n+1} \leq t, E_{n+1} = \tilde{e}\lvert E_n= {e}, A_n = {\alpha}),
\quad\text{ for any }t\geq0, \tilde{e} \in \Ee, ({e}, {\alpha})\in \Kk, n\in\NN.
\end{equation}
To simplify further calculations, we now factorise the conditional probability 
in~\eqref{eq:semiMarkovkernel}.
\begin{proposition}
For any $\tilde{e} := \left(\tilde{j}, \tilde{v}^b, \tilde{v}^a, \tilde{p}, \tilde{z}, \tilde{y}\right)\in \Ee$
and $\left({e} := (j, v^b, v^a, p, z, y), {\alpha} := (m, l)\right)\in \Kk$,
we have
\begin{equation}\label{eq:semiMarkovkernelsimp}
Q(t, \tilde{e}\lvert ({e}, {\alpha}))  = 
{\Qq_{j, v, \alpha}}\left(t, \tilde{j}, \tilde{z}\right)
f_{\tilde{j}}\left(\tilde{v}^b, \tilde{v}^a\right)
\ind_{\{\tilde{p} = {p} + \tilde{j}\}}
\ind_{\{\tilde{y} = {y} - {m} - \tilde{z}\}},
\end{equation}
for all $t\geq 0$, where\footnote{\,$\PP$ (short for~$\PP^{\pi}_{(\ee, \uplambda)}$ in this section) is the probability measure introduced in Theorem~\ref{thm:Tulcea},
and we use the short-hand notation
$\PP\left(X_{n+1}\leq t, J_{n+1} = \tilde{j}, Z_{n+1} = \tilde{z}, \Big\lvert\cdots\right)
 = \PP\left(X_{n+1}\leq t, J_{n+1} = \tilde{j}, Z_{n+1} = \tilde{z}, V^b_{n+1}\in\NN^+, V^a_{n+1}\in\NN^+, P_{n+1}\in\NN^+, Y_{n+1}\in\NN \Big\lvert\cdots\right)$}
 for any~$n\in \NN$,
$$
{\Qq_{j, v, \alpha}}\left(t, \tilde{j}, \tilde{z}\right) :=
\PP\left(X_{n+1}\leq t, J_{n+1} = \tilde{j}, Z_{n+1} = \tilde{z} \Big\lvert 
J_n = {j}, (V^b_n, V^a_n) = ({v}^b, {v}^a), A_n = \alpha\right).
$$
\end{proposition}
%%%%%%%%%%%%%%%%%%%%%%%%%%%%%
\begin{proof}
According to Assumption~\ref{ass:LOBEvolution} and Remark~\ref{rmk:evolution_E}, we can write
\begin{align*}
Q(t, \tilde{e}\lvert ({e}, {\alpha}))
& = \PP\left(X_{n+1}\leq t, J_{n+1} = \tilde{j}, Z_{n+1} = \tilde{z} \lvert E_n = e, A_n = \alpha\right)\times\\
&\qquad \PP\left( (V^b_{n+1}, V^a_{n+1}) = (\tilde{v}^b, \tilde{v}^a), P_{n+1} = \tilde{p}, Y_{n+1} = \tilde{y}\big\lvert
X_{n+1}\leq t, J_{n+1} = \tilde{j}, Z_{n+1} = \tilde{z}, E_n = e, A_n = \alpha\right)\\
& = {\Qq_{j, v, \alpha}}\left(t, \tilde{j}, \tilde{z}\right)
\PP\left((V_{n+1}^b, V_{n+1}^a) = (\tilde{v}^b, \tilde{v}^a)\lvert J_{n+1} = \tilde{j}\right)\times\\
&\qquad
\PP\left(P_{n+1} = \tilde{p}\lvert J_{n+1} = \tilde{j}, P_{n} = {p}\right)
\PP\left(Y_{n+1} = \tilde{y}\lvert Y_n = {y}, M_n = {m}, Z_{n+1} = \tilde{z}\right)\\
& = {\Qq_{j, v, \alpha}}\left(t, \tilde{j}, \tilde{z}\right)
f_{\tilde{j}}\left(\tilde{v}^b, \tilde{v}^a\right)
\ind_{\{\tilde{p} = {p} + \tilde{j}\}}
\ind_{\{\tilde{y} = {y} - {m} - \tilde{z}\}}.
\end{align*}
\end{proof}

%%%%%%%%%%%%%%%%%%%%%%%%%%%%%%%%%%%%%%
\begin{remark}\label{rem:ddotQsemiMarkov}
The function~${\Qq}$ is a semi-Markov kernel on~$\RR^+_0\times \Ee^\prime$ given~$\Kk^\prime$, 
where 
\begin{align*}
\Ee^\prime &:= \{-1, +1\}\times\{0, 1, \dots, N\};\\
\Kk^\prime &:= \left\{(j, v^b, v^a, \alpha): j\in\{+1, -1\}, (v^b, v^a)\in\{1, \dots, N\}^2, \alpha\in\Aa, m < v^b\right\}.
\end{align*}
Indeed, for any~$(j, v^b, v^a, \alpha)\in \Kk^\prime$,
the probability ${\Qq}_{j, v, \alpha}\left(t, \{+1, -1\}, \{0, \dots, {l}\}\right)$ converges to~$1$ 
for large~$t$,
indicating the amount of the matched limit order cannot exceed that of the limit order posted by the agent.
\end{remark}

According to Assumptions~\ref{ass:LOBEvolution},~\ref{ass:PoissonOrderFlow} and~\ref{ass:AdmissibleTradingS},
the semi-Markov kernel~${\Qq}$ describes the dynamical mechanism of a {queueing race} between the volumes sitting at the best bid and ask prices.
Intuitively, fix~$\left(j, v^b, v^a, \alpha\right)\in \Kk^\prime$,
and consider a queueing race starting with~$v^b$ and~$v^a$ units limit orders 
(from the general market participants) 
at the best bid and ask prices at a certain decision epoch.
The agent subsequently submits a sell market order of~${m}$ unit size,
which decreases the best bid volume to~$({v}^b - {m})$ unit size,
and posts a sell limit order of~${l}$ unit size,
which has less time priority than the pre-existing~${v}^a$ units limit orders at the best ask price. 
After the agent's action, mutually independent order book events happen at exponential times with the rates depending on the price move direction~${j}$ and therefore change the volumes of the best bid and ask queues. 
The queueing race terminates whenever the volume of either the best bid or ask queue reaches zero,
and we denote the result of a queueing race by~$+1$ (resp.~$-1$) if the best ask (resp. bid) queue is depleted first.
For~$(t, \tilde{j}, \tilde{z})\in\RR^+_0\times \Ee^\prime$, 
the quantity~${\Qq}_{j, v, \alpha}(t, \tilde{j}, \tilde{z})$
is the probability that the duration of the race is less than or equal to~$t$,
the result is~$\tilde{j}$,
and~$\tilde{z}$ unit size of the agent's limit order gets executed. 
In the following,
we model the dynamics of the volumes at the best bid and ask prices 
as generalised birth-death processes, 
and therefore build a connection between the semi-Markov kernel and the queueing theory. 

%%%%%%%%%%%%%%%%%%%%%%%%%%%%%%%%%%%%%

\begin{definition}\label{def:bdprocesses}
Let $(\overline{\Omega}, \overline{\Ff}, \overline{\PP})$ be a new filtered probability space.
For~$v\in\NN^+$, $l\in\NN$ and~$\kappa, \mu, \theta,\eta >0$, define the following processes on
this space:
\begin{itemize}
\item $\left(B[v, \kappa, \mu, \theta]_s\right)_{s\geq 0}$
is a birth and death process with state space~$\NN$ and absorbing state~$0$, 
given the initial state~$v$;
$\kappa$ is the birth rate and~$\mu+i\theta$ the death rate when in state~$i\in\NN^+$;
\item $\left(C[v, l, \mu, \theta]_s\right)_{s\geq0}$ 
is a pure death process with state space~$\NN$ and absorbing state~$0$
given initial state~$l+v$;
the death rate equals to~$\mu+\max(0, i-l)\theta$ when in state~$i \in\NN^+$;
\item$\left(G[\kappa, \mu, \theta, \eta]_s\right)_{s\geq0}$ is a process with state space~$\NN$
given initial state~$0$. 
Strictly before time~$\eta$,
it is a birth and death process with birth rate~$\kappa$ and death rate~$i\theta$ when in state~$i\in\NN$.
After~$\eta$,
the birth and death rate of this process change to~$\kappa$ and~$\mu + i\theta$ when in state~$i\in\NN^+$ and~$0$ becomes the absorbing state.
\item
$\left(A[v, l, \kappa, \mu, \theta]_s\right)_{s\geq0}$ is a process with state space~$\NN^2$ defined by
$$
A[v, l, \kappa, \mu, \theta]_s := \big(C[v, l, \mu, \theta]_s, G[\kappa, \mu, \theta, \tau_{C[v, l, \mu, \theta]}]_s\big),
\quad\textrm{for $s\geq0$}.
$$
\end{itemize}
\end{definition}

%%%%%%%%%%%%%%%%%%%%%%%%%%%%%%%%%%%%%%%%%%%%%%%%%%%%%%
\begin{lemma}\label{Lem:SMK&BDP}\cite[Lemma~2]{cont2010stochastic}
Fix~$\left(j, v^b, v^a, \alpha\right)\in \Kk^\prime$.
Suppose that, at the~$n$-th decision epoch, 
the queueing race starts with~$v^b$ and~$v^a$ units limit orders at the best bid and ask prices 
after the price moves by~$j$ tick, and the agent takes an action~$\alpha = (m, l)$.
On $[\tau_n, \tau_{n+1})$, define the following processes:
\begin{itemize}
\item $\widetilde{B}$: size of the orders sitting at the best bid price;
\item $\widetilde{C}$: size of the agent's limit order together with the orders with higher time priority at the best ask price;
\item $\widetilde{G}$: size of the orders with lower time priority than the agent's limit order at the best ask price.
\end{itemize}

Then there exist two independent processes~$B[{v}^b-{m}, \kappa_{{j}}^b, \mu_{{j}}^b, \theta_{{j}}^b]$
and~$A[{v}^a, {l}, \kappa_{{j}}^a, \mu_{{j}}^a, \theta_{{j}}^a]$
such that
$$
B[{v}^b-{m}, \kappa_{{j}}^b, \mu_{{j}}^b, \theta_{{j}}^b]_s = \widetilde{B}_{s+\tau_n}
\qquad\text{and}\qquad
A[{v}^a, {l}, \kappa_{{j}}^a, \mu_{{j}}^a, \theta_{{j}}^a]_s = (\widetilde{C}_{s+\tau_n}, \widetilde{G}_{s+\tau_n}),
\qquad\text{for all }s\in[0, \tau_{n+1} - \tau_n).
$$
\end{lemma}
%%%%%%%%%%%%%%%%%%%%%%%%%%%%%%%%%%%%%%%%%%%%%%%%%%%%
According to Lemma~\ref{Lem:SMK&BDP}, we now provide an expression for~${\Qq}$, 
and defer its proof to Appendix~\ref{app:propqbarcal}.
We recall that, for a continuous-time process~$L$, 
the functions~$f_L$ and~$F_L$ are defined in the Notations part, just before Section~\ref{sec:stylisedLOB}.
\begin{proposition}\label{prop:qbarcal}
Fix~$(j, v^b, v^a, \alpha)\in \Kk^\prime$, 
introduce the short-hand notations:
\begin{equation*}
\begin{array}{lrlllrl}
 & \displaystyle {B}^b & := & B[{v}^b - {m}, \kappa^b_{{j}}, \mu^b_{{j}}, \theta^b_{{j}}],
 & \qquad \displaystyle {B}^a & := & B[{v}^a, \kappa^a_{{j}}, \mu^a_{{j}}, \theta^a_{{j}}],\\
 & \displaystyle {A}^{l} & := & A[{v}^a, {l}, \kappa_{{j}}^a, \mu_{{j}}^a, \theta_{{j}}^a],
 & \qquad \displaystyle {C}^{l} & := & C[{v}^a, l, \mu_{{j}}^a, \theta_{{j}}^a],
\end{array}
\end{equation*}
as well as the scenarios:
\begin{table}[!htp]
\centering
\begin{tabular}{|c|c|c|c|c|c|}
\hline
S1 & S2$\pm$ & S3 & S4 & S5 & S6\\
\hline
${l}\geq 1$ & $l=0$ & $l\geq 1$ & $l = 1$ & $l > 1$ & $l>1$\\
$ \tilde{j} = +1$ & $\tilde{j} = \pm 1$ & $\tilde{j} = -1$ & $\tilde{j} = -1$ & $\tilde{j} = -1$ & $ \tilde{j} = -1$\\
& &$ \tilde{z} =0$ & $\tilde{z} = 1$ & $\tilde{z} \in \{1, \dots, {l} -1\}$ & $\tilde{z} = {l}$\\
\hline
\end{tabular}
\end{table}

\vspace{0.1cm}

Then the following holds for any~$(t, \tilde{j}, \tilde{z})\in\RR^+_0\times \Ee^\prime$:
\begin{equation*}
\Qq_{j, v, \alpha}\left(t, \tilde{j}, \tilde{z}\right)
 = 
\left\{
\begin{array}{ll}
\displaystyle\left[F_{{A}^{l}}(t) - \int_0^t f_{{A}^{l}}(u) F_{{B}^b}(u) \D u\right]
\ind_{\{\tilde{z} = {l}\}}, & [S1],\\
\displaystyle \left[F_{{B}^a}(t) - \int_0^t f_{{B}^a}(u)F_{{B}^b}(u)\D u\right]
\ind_{\{\tilde{z} = 0\}}, 
&  [S2+],\\
\displaystyle \left[F_{{B}^b}(t) - \displaystyle \int_0^t f_{{B}^b}(u) F_{{B}^a}(u)\D u\right]
\ind_{\{\tilde{z} = 0\}}, 
&   [S2-],\\
\displaystyle F_{{B}^b}(t) - \int_0^t f_{{B}^b}(u) F_{{C}^{1}}(u) \D u, &  [S3],\\
\displaystyle \int_0^t f_{{B}^b}(u)\left[F_{{C}^{1}}(u) - F_{{A}^{1}}(u)\right]\D u, &  [S4],\\
\displaystyle \int_0^t f^*_{B^b}(\epsilon)\int_0^\epsilon f^*_{C^{\tilde{z}}}(u)\D u\, \D\epsilon, &  [S5],\\
\displaystyle\int_0^t f_{{B}^b}(u)\left[F_{C^1}(u) - F_{A^l}(u)\right]\D u-
\displaystyle\sum_{z = 1}^{{l}-1}
\int_0^t f^*_{B^b}(\epsilon)\int_0^\epsilon f^*_{C^{z}}(u)\D u\, \D\epsilon, &  [S6],\\
0,&\text{otherwise},
\end{array}
\right.
\end{equation*}
where 
$f^{*}_{{C}^{{z}}}(\xi) := \E^{\mu_{{j}}^a \xi}f_{{C}^{{z}}}(\xi)$ 
and 
$f^{*}_{{B}^b}(\xi) := \E^{-\mu_{{j}}^a \xi}f_{{B}^b}(\xi)$
for any~$\xi\geq 0$ and~$z\in\NN^+$.
\end{proposition}
%\begin{remark}
%The quantities~${\Qq}_{j, v, \alpha}(t, -1, 0)$ 
%and ${\Qq}_{j, v, \alpha}(t, -1, \tilde{z})$ in~[S6]
%are computed using~[S3],~[S5].
%\end{remark}
%%%%%%%%%%%%%%%%%%%%%%%%%%%%%%%%%%%%%%%%%%%%%%%%
\subsubsection{Terminal kernel}
According to Theorem~\ref{thm:Tulcea} and the Markovian property~\textbf{(M$1$)},
we can express the terminal kernel as
\begin{equation}
P(\zz\lvert (e, \alpha), \lambda) = 
\PP(X_{n+1} > \lambda, \ZZ = \zz\lvert E_n = e, A_n = \alpha),
\end{equation}
for any~$(e, \alpha)\in \Kk$, $\lambda\in\RR$, $\zz\in\NN$.
Remark~\ref{rmk: tk} implies that only the cases when~$\lambda >0$ and~$\zz \in \{0, \dots, l \}$
need to be considered.
According to Lemma~\ref{Lem:SMK&BDP}, we now provide an expression for~${\Qq}$, 
proved in Appendix~\ref{app:Pcal}.

\begin{proposition}\label{prop:Pcal}
For any~$\lambda>0, (e, \gamma)\in \Kk$ (with corresponding~$(j, v^b, v^a, m, l)\in \Kk^\prime$)), 
introduce the processes~$B^b, B^a, A^l, C^l$ as in Proposition~\ref{prop:qbarcal}.
Then the following equality holds:
\begin{equation*}
P(\zz\lvert (e, \gamma), \lambda)
 = 
 \left\{
 \begin{array}{ll}
\overline{F}_{B^b}(\lambda)\overline{F}_{B^a}(\lambda), & \text{if }l=0\text{ and }\zz = 0,\\
\overline{F}_{B^b}(\lambda)\overline{F}_{C^1}(\lambda), & \text{if }l\geq1\text{ and }\zz = 0,\\
\overline{F}_{B^b}(\lambda)\left[F_{C^\zz}(\lambda) - \left(F_{C^\zz} \ast F_{\Xi}\right)(\lambda)\right], & \text{if }l>1\text{ and }\zz\in \{1, \dots, l -1\},\\
\overline{F}_{B^b}(\lambda)\left[F_{C^l}(\lambda) - F_{A^l}(\lambda)\right], & \text{if }l\geq1\text{ and }\zz = l,\\
0, & \text{otherwise},
 \end{array}
\right.
\end{equation*}
where~$\Xi$ is an exponentially distributed random variable with parameter~$\mu_j^a$, 
and~$\ast$ is the convolution operator.
\end{proposition}

%%%%%%%%%%%%%%%%%%%%%%%%%%%%%%%%%%%%%%%%%%%%%%%%
%%%%%%%%%%%%%%%%%%%%%%%%%%%%%%%%%%%%%%%%%%%%%%%%
\subsection{Laplace method}\label{sec:queuedistribution}
Not surprisingly, the distributions of the first-passage time of the generalised birth-death processes 
$A, B, C$ in Definition~\ref{def:bdprocesses}
do not admit closed-form expressions.
To compute them, we first determine their Laplace transforms, and invert them numerically.
We keep here the notations of Proposition~\ref{prop:qbarcal}.
\begin{definition}
Let~$f: \RR^+_0 \to \RR$ be a function absolutely integrable on~$[0, \omega]$ for any~$\omega>0$.
Its (one-sided) Laplace transform is defined by
$\hat{f}(s) := \displaystyle\lim_{\omega\uparrow\infty}\int_0^{\omega}\E^{-st}f(t)\D t$,
for all $s\in\CC$ such that the right-hand side converges.
\end{definition}
The standard (albeit simplified) inversion formula for the Laplace transform is the Bromwich contour integral, or Mellin inversion~\cite[Chapter 1]{abate2000introduction}:
for an absolutely integrable continuous function~$f$, 
the identity 
$f(t) = \frac{1}{2\pi\I}\int_{x -\I\infty}^{x + \I\infty}\E^{ts}\hat{f}(s)\D s$
holds for any~$x>0$, and, by symmetry arguments, can be simplified to
\begin{equation}\label{eq:inverseLT}
f(t) =\frac{2\E^{xt}}{\pi}\int_0^{\infty} \Re\left[\hat{f}(x+\I u)\right]\cos(ut)\D u,
\qquad\text{for all }t>0.
\end{equation}
We then apply the Euler algorithm~\cite[Section 1]{abate1995numerical}
that exploits the specific structure of the integrand in~\eqref{eq:inverseLT}.
We now consider the general case of a birth-death process~$X^b$
with initial state~$b\in\NN^+$,
and
with birth rate $\lambda_n\geq0$ and death rate $\mu_n>0$ in state $n \in\NN^+$.
%and let $\tau_b$~($b\in\NN^+$) denote its first-passage time to the origin 
%given the initial state~$b$. 
The following lemma, derived in~\cite[Equation (14)]{cont2010stochastic} following Abate-Whitt methodology~\cite[Section 4]{abate1999Computing}, 
expresses the Laplace transforms of the density and cumulative distribution function of~$\tau_{X^b}$.
\begin{lemma}\label{lem:LaplaceMethodf}
The equality $\hat{F}_{X^b}(s) = s^{-1}\hat{f}_{X^b}(s)$ holds on $\{s\in\CC: \Re(s)>0\}$, 
and 
\begin{equation}\label{eq:LTpdfcf}
\hat{f}_{X^b}(s) = \prod_{n=1}^b\left[-\frac{1}{\lambda_{n-1}}\foo_{k\geq 0}
\left(\frac{-\lambda_{k+n-1}\mu_{k+n}}{\lambda_{k+n} + \mu_{k+n} +s}\right)\right],
\quad\text{for all }s \in \CC \text{ such that }\Re(s)>0,
\end{equation}
where
$\displaystyle
\foo_{k\geq 0}\frac{a_k}{b_k} := \lim_{k\uparrow\infty}t_0\circ t_1\circ\dots\circ t_k(0)$
and 
$\displaystyle
t_k(u) := \frac{a_k}{b_k + u}$ for $u\geq 0$.
\end{lemma}
%%%%%%%%%%%%%%%%%%%%%%%%%%%%%%%%%%%%%%%%%%%%%%%%
%%%%%%%%%%%%%%%%%%%%%%%%%%%%%%%%%%%%%%%%%%%%%%%%

%%%%%%%%%
\begin{proposition}\label{prop:Laplace}
Fix $v\in\NN^+$, $l\in\NN$ and $\kappa, \mu, \theta >0$, and denote the processes
$A[v, l, \kappa, \mu, \theta]$, $B[v, \kappa, \mu, \theta]$, $C[v, l, \mu, \theta]$ 
(as in Definition~\ref{def:bdprocesses})
by~$A$, $B$ and~$C$, respectively.
In particular, we denote the process~$B[j, \kappa, \mu, \theta]$ by~$B^j$ for any~$j\in\NN^+$.
Assume that~$f_{A}$, $f_{B}$ and~$f_{C}$ are continuous on~$\RR^+$.
Then 
$$
\hat{f}_{{B}}(s) = \frac{1}{(-\kappa)^{v}}\prod_{n = 1}^v
\foo_{k\geq 0}
\left[\frac{-\kappa\mu-\kappa(k+n)\theta}{\kappa + \mu + (k+n)\theta + s}\right],
\quad\text{and}\quad
\hat{f}_{C}(s) = \left(\frac{\mu}{\mu+s}\right)^l
\prod_{n = l+1}^{l+v}\frac{\mu + (n - l)\theta}{\mu + (n-l)\theta + s},
$$
for~$\Re(s)>0$.
Besides, given
$
R_j(u) := \frac{1}{j!}
\exp\left(   -\frac{\kappa}{\theta}\left(1 - \E^{-\theta u}\right)   \right)
\left[\displaystyle\frac{\kappa}{\theta}\left(1 - \E^{-\theta u}\right) \right]^j$
for $u\geq 0$ and $j\in\NN$,
we have
\begin{equation}\label{eq:LTfA}
f_{A}(t) = f_{C}(t)R_0(t) + \int_0^t\sum_{j=1}^{\infty}f_{B^j}(t - u)f_{C}(u) R_j(u)\D u,\quad\text{for all }t\geq0.
\end{equation}
\end{proposition}
\begin{proof}
The formulae for~$\hat{f}_B$ and~$\hat{f}_C$
are derived directly from Lemma~\ref{lem:LaplaceMethodf}, 
and we therefore focus on~\eqref{eq:LTfA}.
Let~$\tau_\Delta := \tau_A - \tau_{C}$.
Before time~$\tau_{C}$, the process~$(G_u) := (G[\kappa, \mu, \theta, \tau_C]_u)$ 
can be regarded as an initial empty~$M/M/\infty$ queue with arrival rate~$\kappa$ 
and service rate~$\theta$.
Let~$R_j(u)$ denote the probability of~$G_u$ being in state~$j\in\NN$
when~$u<\tau_C$.
Then, by~\cite[p. 160]{takacs1959introduction}, we have
\begin{equation*}
R_j(u) = \overline{\PP}\left(G_u = j \Big\lvert\tau_C = u\right)
= \frac{1}{j!}
\exp\left\{-\frac{\kappa}{\theta}\left(1 - \E^{-\theta u}\right)\right\}
\left[ \displaystyle\frac{\kappa}{\theta}\left(1 - \E^{-\theta u}\right) \right]^j.
\end{equation*}
Given~$\tau_C = u, G_u = j\in\NN^+$, the probability density function of~$\tau_\Delta$ is~$f_{B^j}$.
Indeed, in the case when $\tau_C = u$ and~$G_u = G_{\tau_C} = j$,
the time spent on depleting the agent's order and the orders with higher time priority is~$u$
and at that time the volume remaining in the queue is of~$j$ unit size.
The remaining queue can be described by the process~$B^j$, and the depletion time~$\tau_\Delta$ is thus~$\tau_{B^j}$ 
(with density~$f_{B^j}$).
And given~$\tau_C = u, G_u = 0$, we have~$\Delta = 0$ almost surely. 
Therefore, the mixture density 
$\delta(\cdot)R_0(u) + \sum_{j=1}^{\infty}f_{B^j}(\cdot)R_j(u)$,
with~$\delta(\cdot)$ being the Dirac mass, 
provides the density of~$\tau_\Delta$
given $\tau_C = u$.
Furthermore, the function
$\delta(\cdot - u)R_0(u) + \sum_{j=1}^{\infty}f_{B^j}(\cdot - u)R_j(u)(u)$
is the density of $\tau_A = \tau_\Delta + \tau_{C}$ given $\tau_{C} = u$.
Consequently, we obtain~\eqref{eq:LTfA}.
\end{proof}
%\begin{remark}
%The expressions for the Laplace transforms of~${F}_B$ and of~${F}_C$ 
%can be obtained directly using Proposition~\ref{prop:Laplace}
%and Lemma~\ref{lem:LaplaceMethodf};
%integrating~\eqref{eq:LTfA}, we obtain, for all $t>0$,
%$$
%F_A(t) = \int_0^t f_C(u)R_0(u)\D u + \int_0^t \sum_{j = 1}^{\infty}F_{B^j}(t - u)f_C(y)R_j(u)\D u.
%$$
%\end{remark}

%%%%%%%%%%%%%%%%%%%%%%%%%%%%%%%%%%%%%%%%%%%%%%%%
%%%%%%%%%%%%%%%%%%%%%%%%%%%%%%%%%%%%%%%%%%%%%%%%
\section{Existence of Optimal Policy}\label{Section:OptimalStrategy}
We now illustrate our main result, 
namely the existence and uniqueness of the value function,
and the existence of a stationary optimal policy.
\begin{theorem}\label{thm:ValueFunction}
The value function~$V^*$ in~\eqref{eq:VF} exists and is unique,
and
there exists a stationary~$\TT$-optimal policy~$\pi^{\phi^*} := \{\phi^*, \phi^*, \dots\}\in\Pi^S$
in~\eqref{eq:optimalpolicy}, with, for any~$(e, \lambda)\in\Ee\times\TT$,
\begin{equation}\label{eq:OptimalSPolicy}
\phi^*(e, \lambda) = \underset{\alpha\in\Aa(e)}{\arg\max}
\left\{
r(e,\alpha) + \sum_{\zz = 0}^\infty w(e, \alpha, \zz)P(\zz\lvert (e, \alpha), \lambda)
+\sum_{\tilde{e}\in \Ee}\int_0^\lambda V^*(\tilde{e}, \lambda - t)Q(\D t, \tilde{e}\lvert (e, \alpha))
\right\}.
\end{equation}
\end{theorem}

The proof of Theorem~\ref{thm:ValueFunction} relies on several ingredients.
First, for a finite-horizon semi-Markov decision model to be sensible,
it is essential to have a (almost surely) finite number of decision epochs before maturity. 
In our setting, this is equivalent to the following lemma.
\begin{lemma}\label{lem:fntde}
For any~$(\ee, \uplambda)\in \Ee\times\TT$, $\pi\in\Pi$, the limit
$\lim\limits_{n\uparrow \infty}\PP^\pi_{(\ee, \uplambda)}(\Nn < n) = 1$
holds for~$\Nn$ as in Definition~\ref{def: measurablespace}.
\end{lemma}
\begin{proof}
According to~\cite[Proposition 2.1]{huang2011finite}, 
it suffices to prove that there exist $\zeta, \upsilon>0$ such that
\begin{equation}\label{ineq:regu1}
{Q}(\zeta, \Ee\lvert ({e}, {\alpha})) \leq 1- \upsilon,
\end{equation}
for any~$({e}, {\alpha})\in \Kk$.
According to~\eqref{eq:semiMarkovkernel} and Lemma~\ref{Lem:SMK&BDP}, we can write,
for any $\zeta>0$ and~$({e}, {\alpha})\in \Kk$,
$$
{Q}(\zeta, \Ee\lvert({e}, {\alpha})) 
= \PP^\pi_{(\ee, {\uplambda})}(X_{n+1}\leq \zeta\lvert E_n = {e}, A_n = {\alpha})
 = \overline{\PP}\left(\tau_{{B}^b}\wedge\tau_{{A}^l}\leq \zeta\right)
= 1 - \overline{\PP}\left(\tau_{{B}^b} > \zeta\right) \overline{\PP}\left(\tau_{{A}^l} > \zeta\right).
$$
By Assumption~\ref{ass:AdmissibleTradingS}\eqref{itm:Strategysecond}, the agent never consumes up all the volumes at the best bid price through submitting market orders, 
so that there is at least one unit size order left at the best bid and ask price after the agent's action.
Then according to stochastic ordering for the birth and death processes~\cite[Section 3]{irle2003stochastic}, the inequalities
$$
\overline{\PP}\left(\tau_{{B}^b}> \zeta\right)
\geq
\overline{\PP}\left(\tau_{B[1, 0, \mu_{{j}}^b, \theta_{{j}}^b]} > \zeta\right) 
\geq \E^{-{\iota}\zeta},
$$
and
$$
\overline{\PP}\left(\tau_{{A}^l} > \zeta\right)
\geq
\overline{\PP}\left(\tau_{C^l} > \zeta\right)
\geq
\overline{\PP}\left(\tau_{C[1, 0, \mu_{{j}}^a, \theta_{{j}}^a]} > \zeta\right)
\geq \E^{-{\iota}\zeta},
$$
hold with
%\begin{equation}\label{eq:iota}
$\iota := \max\left\{\mu_j^\ssf+\theta_j^\ssf: (\ssf, j) \in \{a, b\}\times\{+1, -1\}\right\}$,
%\end{equation}
and~\eqref{ineq:regu1} therefore holds for $\zeta>0$ and~$\upsilon = \E^{-2{\iota}\zeta}$.
\end{proof}
Next, let~$\Uu$ denote the Banach space of non-negative valued functions on~$\Ee\times\TT$ with a finite supremum norm:
$$
\Uu:=\left\{u:\Ee\times\TT\to\RR_0^+\,\Bigg\lvert\,
\|u\|:= \sup_{(e,\lambda)\in \Ee\times\TT}|u(e, \lambda)|<\infty\right\}.
$$
and, for any decision rule~$\phi\in\Phi$,
introduce the dynamic programming operator~$\Tt^\phi$ acting on~$\Uu$ as
\begin{align*}
\Tt^\phi u(e, \lambda) 
 & :=
 r(e, \phi(e, \lambda)) + \sum_{\zz = 0}^\infty w(e, \phi(e, \lambda), \zz)P(\zz\lvert (e, \phi(e, \lambda)), \lambda)
 + \sum_{\tilde{e}\in \Ee}\int_0^\lambda u(\tilde{e}, \lambda-t)Q\big(\D t, \tilde{e}\lvert (e, \phi(e, \lambda))\big),
\end{align*}
for any~$u\in\Uu$ and~$(e, \lambda)\in \Ee\times \TT$.
The following proposition, as proved in Appendix~\ref{app:PropTpi}, gives properties of~$\Tt^\phi$.

\begin{proposition}\label{prop:PropTphi}
For any~$\phi\in\Phi$ and~$\pi := \{\phi_0, \phi_1, \phi_2, \dots\}\in\Pi$,
the following hold:
\begin{enumerate}[(a)]
\item \label{prop:PropTphifirst}
$\Tt^\phi$ is a monotone contraction on~$\Uu$ with codomain~$\Uu$;
\item \label{prop:PropTphisecond} 
the identity~$V^\pi = \Tt^{\phi_0} V^{\pi_-}$ is valid on~$\Ee\times\TT$,
where~$\pi_- := \{\phi_1, \phi_2, \dots\}\in\Pi$.
\end{enumerate}
\end{proposition}
We can now prove Theorem~\ref{thm:ValueFunction}.

\begin{proof}[Proof of Theorem~\ref{thm:ValueFunction}]
By Proposition~\ref{prop:PropTphi}\eqref{prop:PropTphisecond}, the identity~$V^{\pi^\phi} = \Tt^\phi V^{\pi^\phi}$ holds for any~$\phi\in\Phi$ and corresponding stationary policy~$\pi^\phi := \{\phi, \phi, \dots\}\in \Pi^S$.
For any~$\pi\in\Pi$, the finiteness of the state space and the action space together with Lemma~\ref{lem:fntde} yield that~$V^\pi\in\Uu$.
Therefore, Banach Fixed-Point's Theorem~\cite{granas2013fixed}
and Proposition~\ref{prop:PropTphi}\eqref{prop:PropTphifirst}
guarantee existence and uniqueness of~$V^{\pi^\phi}$ and that
\begin{equation}\label{eq:pfeq1}
\lim_{n\uparrow\infty}\left(\Tt^\phi\right)^n u = V^{\pi^\phi},\quad\text{for any }u\in\Uu.
\end{equation}

Introduce now the iteration operator~$\Vv$ acting on~$\Uu$ as, 
for any $u\in\Uu$ and $(e, \lambda)\in \Ee\times \TT$,
\begin{equation}\label{eq:DPoA}
\Vv u(e,\lambda) := \sup_{\alpha\in\Aa(e)}\left\{
r(e,\alpha) + \sum_{\zz = 0}^\infty w(e, \alpha, \zz)P(\zz\lvert (e, \alpha), \lambda)
+\sum_{\tilde{e}\in \Ee}\int_0^\lambda u(\tilde{e}, \lambda - t)Q(\D t, \tilde{e}\lvert (e, \alpha))
\right\},
\end{equation}
which is also a contraction with codomain~$\Uu$. 
Indeed, $\Vv(\Uu)\subset\Uu$ is immediate since the action space is finite,
and the contraction property is inherited from that of~$\Tt^\phi$ 
by~\cite[Theorem 2]{denardo1967contraction}.
Banach Fixed-Point's Theorem~\cite{granas2013fixed} then ensures that~$\Vv u = u$ has a unique solution, denoted by~$u^*$.
By~\cite[Section 1]{Mamer1986successive}, 
the fixed point~$u^*$ admits a maximiser~$\phi^*$ such that~$u^* = \Tt^{\phi^*}u^*$, 
with, for any~$(e, \lambda)\in\Ee\times\TT$,
\begin{equation}\label{eq:Toptimalp}
\phi^*(e, \lambda) := \underset{\alpha\in\Aa(e)}{\arg\max}
\left\{
r(e,\alpha) + \sum_{\zz = 0}^\infty w(e, \alpha, \zz)P(\zz\lvert (e, \alpha), \lambda)
+\sum_{\tilde{e}\in \Ee}\int_0^\lambda u^*(\tilde{e}, \lambda - t)Q(\D t, \tilde{e}\lvert (e, \alpha))
\right\}.
\end{equation}
Suppose now that a policy~$\pi^*:= \left\{\phi^*_0, \phi^*_1, \phi^*_2, \dots\right\}\in\Pi$ is~$\TT$-optimal in~\eqref{eq:optimalpolicy}.
Proposition~\ref{prop:PropTphi} and~\eqref{eq:pfeq1} yield
\begin{equation}\label{eq:pfeq2}
V^* = V^{\pi^*} = \Tt^{\phi^*_0}V^{\pi^*_-} \leq \Tt^{\phi^*_0} V^{\pi^*}\leq\lim_{n\uparrow\infty}\left(\Tt^{\phi^*_0}\right)^nV^{\pi^*} = V^{\pi^{\phi^*_0}}.
\end{equation}
Combining this with~$V^{\pi^{\phi^*_0}}\leq V^*$ by Definition~\ref{Def:expectedReward} indicates that
the stationary policy~$\pi^{\phi^*_0} := \left\{\phi^*_0, \phi^*_0, \dots\right\}\in\Pi^S$ is also~$\TT$-optimal.
Since~$u^* = \Tt^{\phi^*}u^*\geq \Tt^{\phi^*_0}u^*$, 
applying Proposition~\ref{prop:PropTphi} and~\eqref{eq:pfeq1} we obtain
$$
V^* = V^{\pi^{\phi^*_0}}=\lim_{n\uparrow\infty}\left(\Tt^{\phi^*_0}\right)^nu^*\leq\Tt^{\phi^*_0}u^*\leq u^* = \Tt^{\phi^*}u^*  =\lim_{n\uparrow\infty} \left(\Tt^{\phi^*}\right)^n u^* = V^{\pi^{\phi^*}}\leq V^*,
$$
and Theorem~\ref{thm:ValueFunction} follows. 
\end{proof}

%%%%%%%%%%%%%%%%%%%%%%%%%%%%%%%%%%%%%%%%%%%%%%
%%%%%%%%%%%%%%%%%%%%%%%%%%%%%%%%%%%%%%%%%%%%%%
\section{Empirical studies}\label{sec:empirical}
Our empirical calculations are based on the `Level-I' LOBSTER data
for three large-tick stocks: Microsoft~(MSFT), Intel~(INTC) and Yahoo~(YHOO),
that are traded on the Nasdaq platform from 11 April 2016 to 15 April 2016,
recording all market order arrivals, limit order arrivals, 
and cancellations at the best prices between 9.30am and 4pm.
These three large-tick stocks are selected due to price, trading volume and market share considerations~\cite[Section 4]{bonart2017latency}.
In order to avoid the impact from the abnormal trading behaviours shortly after market opening 
or shortly before market closing,
we exclude market activities during the first and the last twenty minutes of each trading day.
We also exclude all executions of hidden orders which accounts for around~$12\%$ of the entire trading volume.
In the following, we first (Section~\ref{sec:EstMd}) illustrate the estimation methodology
of the Poisson parameters in Assumption~\ref{ass:PoissonOrderFlow},
as well as the joint distribution of the best volumes after a price change in Assumption~\ref{ass:LOBEvolution}\eqref{ass:LOBEvolution3}.
We then (Section~\ref{sec:OptmStrtg}) give a numerical scheme that approximates the value function in~\eqref{eq:VF}.
We finally (Section~\ref{sec:OSresult}) visualise the optimal decision rule in~\eqref{eq:OptimalSPolicy} for liquidating the stock YHOO under different trading conditions. 

%%%%%%%%%%%%%%%%%%%%%%%%%%%%%%%%%%%%%%%%%%%%%%
\subsection{Parameter estimation}\label{sec:EstMd}
\subsubsection{Poisson parameters}
As in Assumption~\ref{ass:Underlying}\eqref{ass:Underlying1},
orders from the general market participants are of unit size. 
We first compute the average size of the limit orders, market orders and cancellations at the best prices, denoted by~$S^l, S^m$ and~$S^c$ respectively, and choose the unit size to be~$S^l$. 
Estimation results are given in Table~\ref{Tab:aos}.
%%%%%%%%%%
\begin{table}[!htp]
\centering
\begin{tabular}{|c|c|c|c|}
\hline
 & \quad MSFT \,\,\,\quad & \quad INTC \,\,\,\quad & \quad YHOO \,\,\,\quad \\
\hline
$S^l$ &176 &317  & 209 \\
\hline
$S^m$ & 332& 565 & 334 \\
\hline
$S^c$ & 163 & 309 & 201 \\
\hline
\end{tabular}
\caption{average order size (in shares)}
\label{Tab:aos}
\end{table}
%%%%%%%%%%

We then estimate the Poisson parameters as follows.
From historical data, we formulate the set~$\Qf_{+1}$ (resp.~$\Qf_{-1}$) as the queueing races happening immediately after a price increase (resp. decrease):
if the spread is currently one tick, 
a queueing race~$\qf_{+1}\in\Qf_{+1}$ (resp.~$\qf_{-1}\in\Qf_{-1}$) starts when the best bid (resp. ask) price increases (resp. decreases) by one tick after the best ask (resp. bid) queue depletes, 
and ends whenever either the new best ask or bid queue depletes. 
By maximum likelihood estimation (see Appendix~\ref{App:MLE}), we have
\begin{equation}\label{eq:MaxLikelihoodResult}
\hat{\mu}^\ssf_{j}  = \frac{N^m_{\ssf, j}}{D_j}\frac{S^m}{S^l},\qquad\qquad
\hat{\kappa}^\ssf_{j}  = \frac{N^l_{\ssf, j}}{D_j},\qquad\qquad
\hat{\theta}^\ssf_{j}  = \frac{N^c_{\ssf, j}}{V_{\ssf, j}}\frac{S^c}{S^l},\qquad\qquad
\text{for }\ssf\in\{a, b\}\text{ and }j\in\{+1, -1\},
\end{equation}
where
\begin{itemize}
\item $N^m_{\ssf, j}, N^l_{\ssf, j}$ and~$N^c_{\ssf, j}$ represent the total number of market orders, limit orders and cancellations at~$\ssf$ price\footnote{By abuse of language, `at~$a$ (resp,~$b$ price)' means `at the best ask (resp. bid) price'.} for the queueing races in set~$\Qf_j$;
\item $D_j$ represents the sum of the length of the queueing races in~$\Qf_j$;
\item $V_{\ssf, j} := \displaystyle\sum_{i=1}^{\# \Qf_j}\int_{\mathfrak{T}_i} \textrm{Vol}^{\ssf}_i(t)\D t$,
where~$\textrm{Vol}^\ssf_i(t)$ (resp. $\mathfrak{T}_i$) denotes
the volume in unit size at~$\ssf$ price at time~$t$
(resp. the time interval) of the $i$-th queuing race in~$\Qf_j$.
\end{itemize}
%%%%%%%%%%
Table~\ref{tab:ParaEst0msLtc} gives the Poisson parameter estimation where the agent's action at each decision epoch has no latency. For the three stocks, we find that:
\begin{itemize}
\item the rates of market order arrivals are indifferent to the side of the best price and the price move direction;
\item immediately after a price increase (resp. decrease), there is a higher rate of limit order arrivals and cancellations at the best bid (resp. ask) price than at the best ask (resp. bid) price;
\item from an estimation (of the Poisson parameters) point of view, an increase of the price on the bid (resp. ask) side is symmetric to a decrease of price on the ask (resp. bid) side.
\end{itemize}
Table~\ref{tab:ParaEst1msLtc} gives the Poisson parameter estimation where the agent's action at each decision epoch has a one-millisecond latency\footnote{When estimating the Poisson parameters in this case, market activities at the first one millisecond of each queueing race are excluded,
and the queueing races with duration shorter than one millisecond are excluded.}. 
By comparing it with Table~\ref{tab:ParaEst0msLtc}, we observe that:
\begin{itemize}
\item the rates of market order arrivals barely change;
\item the rates of limit order arrivals and cancellations see a decrease,
especially on the bid side after a price increase and on the ask side after a price decrease;
\item the symmetry remains unaffected.
\end{itemize}

\begin{table}[!htp]
\centering
\begin{tabular}{|c|c|c|c|c|c|c|c|c|c|c|}
\hline
\multicolumn{2}{|c|}{ } & \multicolumn{3}{|c|}{MSFT}&\multicolumn{3}{|c|}{INTC} & \multicolumn{3}{|c|}{YHOO}\\
\hline
$\ssf$ & $j$ & $\mu$ & $\kappa$ & $\theta$ & $\mu$ & $\kappa$ & $\theta$ &$\mu$ & $\kappa$ & $\theta$\\
\hline
$a$  & +1 & 0.32 &3.07  &0.31       & 0.16 & 2.45 & 0.16       & 0.14 & 1.97 & 0.26\\
$b$  & +1 & 0.34 & 5.97 &0.50       & 0.17 & 3.59  & 0.21       & 0.17 & 3.54  & 0.32\\
$a$  & -1  & 0.35 &5.97  & 0.51      & 0.18 & 3.87 &  0.22     & 0.15 & 3.29 &0.33 \\
$b$  & -1  &  0.34&3.06  &  0.32     & 0.18 & 2.22 &  0.16    & 0.15 & 1.92 &0.21 \\
\hline
\end{tabular}
\caption{Poisson parameter estimation with no latency}
\label{tab:ParaEst0msLtc}
\end{table}

\begin{table}[!htp]
\centering
\begin{tabular}{|c|c|c|c|c|c|c|c|c|c|c|}
\hline
\multicolumn{2}{|c|}{ } & \multicolumn{3}{|c|}{MSFT}&\multicolumn{3}{|c|}{INTC} & \multicolumn{3}{|c|}{YHOO}\\
\hline
$\ssf$ & $j$ & $\mu$ & $\kappa$ & $\theta$ & $\mu$ & $\kappa$ & $\theta$ &$\mu$ & $\kappa$ & $\theta$\\
\hline
$a$  & +1 &0.31 &2.89 &0.27         &0.15&2.36&0.15        &0.13&1.87&0.23\\
$b$  & +1 &0.33 &3.31 &0.40       &0.19&2.46&0.17        &0.16&2.07&0.26\\
$a$  & -1  &0.34 &3.22 &0.41       &0.18&2.49&0.18        &0.14&2.02&0.27\\
$b$  & -1  &0.34 &2.87 &0.27         &0.19&2.36&0.17        &0.15&1.83&0.18\\
\hline
\end{tabular}
\caption{Poisson parameter estimation with 1ms latency}
\label{tab:ParaEst1msLtc}
\end{table}
%%%%%%%%%%%%%%%%%%%%%%%%%%%%%%%%%%%%%%%%%%%%%
%%%%%%%%%%%%%%%%%%%%%%%%%%%%%%%%%%%%%%%%%%%%%
\subsubsection{Volume distribution after a price change}
The volume in unit size is approximated by 
rounding the division of the volume in shares by~$S^l$ up to the nearest integer.
Figure~\ref{fig:YHOOf1msltc} compares the volume distribution 
immediately after a price change and one millisecond later for~YHOO\footnote{For implementing numerical calculation, we introduce the truncation by assuming~$f_{\pm 1}(v^b, v^a) = 0$ for any~$v^b, v^a > 25$ 
since the inequality~$\sum_{v^b = 1}^{25}\sum_{v^a = 1}^{25} f_{\pm 1}(v^b, v^a) \geq 95\%$ holds right after a price change and one millisecond later for YHOO.}.
We observe that:
\begin{itemize}
\item the volume at the best bid (resp. ask) price is quite thin immediately after a price increase (decrease), 
but see a dramatic increase one millisecond later;
\item the volume at the best ask (resp. bid) price keeps the distribution almost unchanged within the first millisecond of the queueing race starting with a price increase (resp. decrease).
\end{itemize}
\begin{figure}[!htp]
\centering
\includegraphics[scale=0.46]{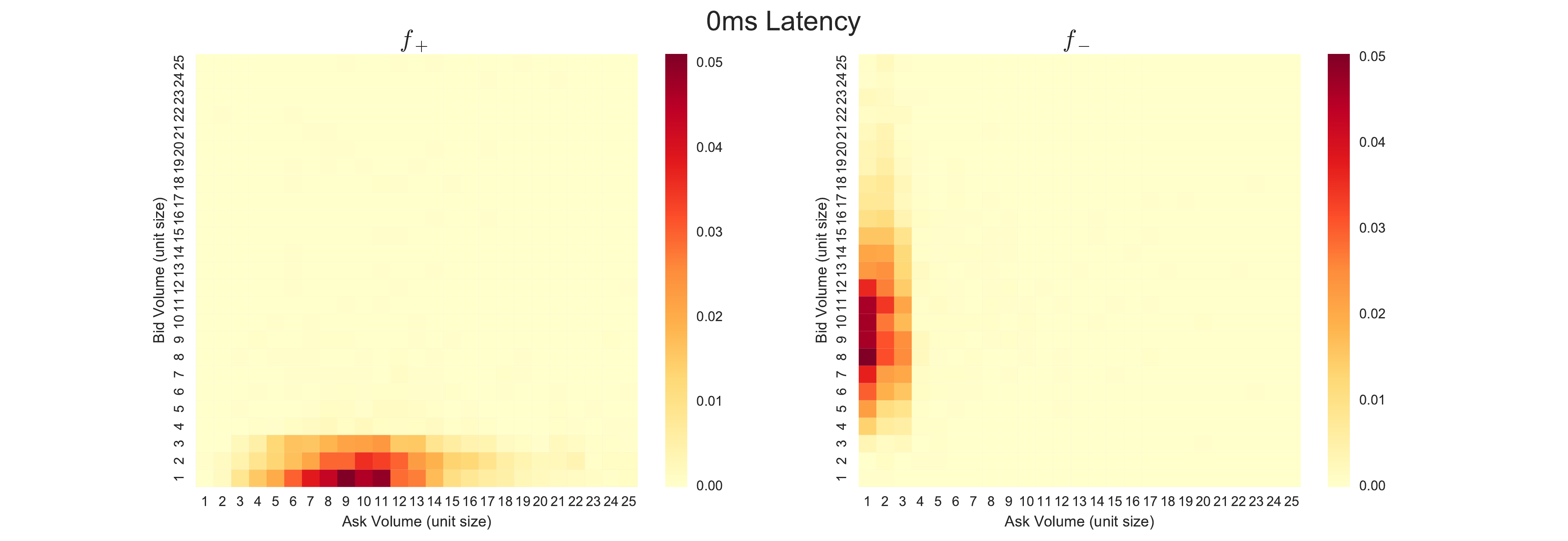}
\includegraphics[scale=0.46]{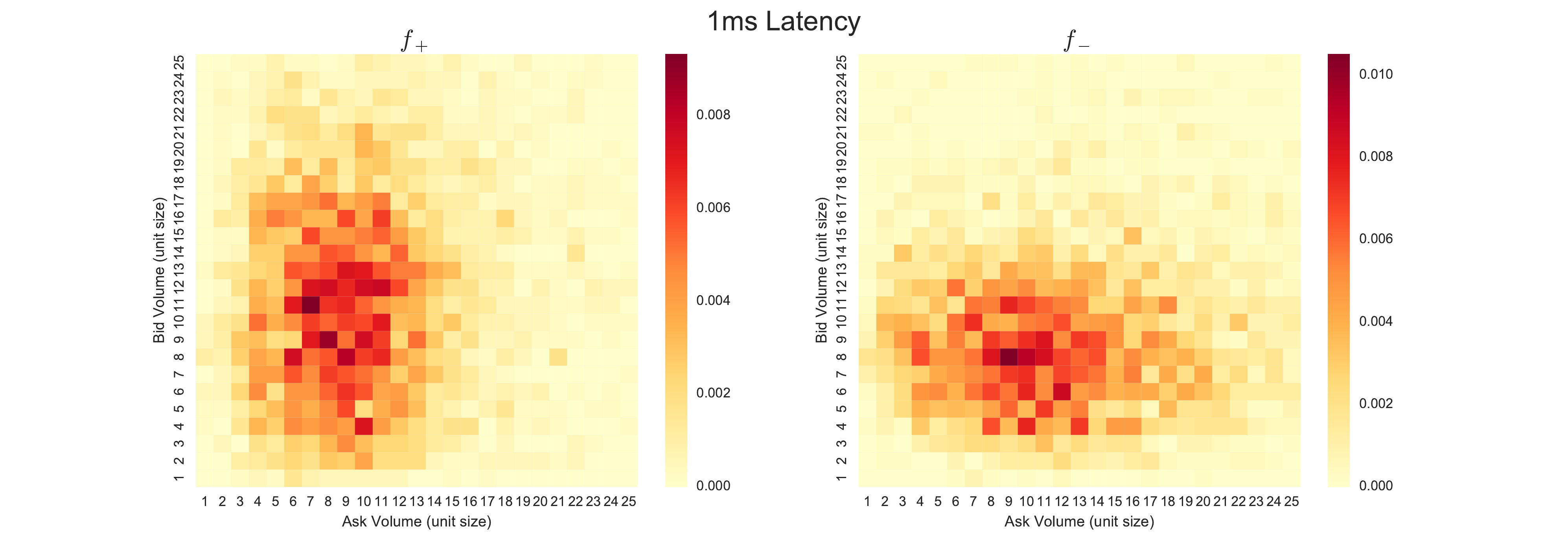}
%\label{fig:YHOOf0msltc}
\caption{YHOO:~$f_{+1}$ (left) and~$f_{-1}$ (right) with no latency (top) and with 
$1$ms latency (bottom)}
\label{fig:YHOOf1msltc}
\end{figure}

\subsection{Numerical scheme}\label{sec:OptmStrtg}
Dynamic programming techniques usually suffer from the `curse of dimensionality'~\cite{powell2007approximate}
to compute the value function through the iteration operator~$\Aa$ in~\eqref{eq:DPoA}.
The next proposition, proved in Appendix~\ref{app:VFrdcfm}, 
allows us to reduce the dimension of the problem, and hence to accelerate the implementation.
%%%%%%%%
\begin{proposition}\label{prop:VFrdcfm}
Given $e := (j, v^b, v^a, p, z, y)\in \Ee$,~$\overline{e} := (j, v^b, v^a, \overline{p}, \overline{z}, y)\in \Ee$ and~$\lambda\in\TT$,
we have
\begin{align*}
V^*(e, \lambda) 
= V^*(\overline{e}, \lambda) + \rho({p} - \overline{p})(y+{z}) + \rho({z} - \overline{z}) (\overline{p} - j).
\end{align*}
\end{proposition}
Besides, the value function~$V^*$ is monotone with respect to time to maturity. 
Indeed, let~$\pi^{*}:=\left\{\phi^*, \phi^*, \dots\right\}$ be~$\TT$-optimal 
and construct a policy~$\pi^{\delta} := \left\{\phi^\delta, \phi^\delta, \dots\right\}$,
for fixed~$\mathcal{\delta}\in(0, T)$, as
\begin{equation*}
\phi^{\delta}(e, \lambda)
:= 
 \left\{
 \begin{array}{ll}
\phi^*(e, \lambda-\delta), & \text{ if }\delta\leq\lambda\leq T,\\
(0, 0), & \text{ if }0\leq\lambda < \delta.
 \end{array}
\right.
\end{equation*}
Definition~\ref{Def:expectedReward} immediately implies 
that $V^{\pi^\delta}(e, \lambda) \leq V^*(e, \lambda)$
for any~$(e, \lambda)\in \Ee\times \TT$ and
$
V^{\pi^\delta}(e, \lambda) = V^*(e, \lambda - \delta)
$
for any~$(e, \lambda)\in \Ee\times [\delta, T]$.

The monotonicity in time to maturity therefore follows since~$\delta$ is arbitrary.
As in~\cite{jiang2015approximate, nascimento2009optimal},
we can take advantage of the monotonicity of the value function to get a faster convergence rate. 
The implementation procedure proceeds as follows, for some tolerance level~$\mathfrak{tol}$:
\begin{enumerate}[Step 1.]
\item (initialization): let~$n = 0$ and~$V_0(e, \lambda) = \rho(p-1)y + \lambda\rho y/T$ for every~$(e, \lambda)\in \Ee\times\TT$;
\item (iteration): choose a random pair~$(e_n, \lambda_n)\in \Ee\times \TT$ and compute~$\widehat{V}_n:=\mathcal{A}V_n(e_n, \lambda_n)$;
\item (correction): with
$\widehat{U}_n := \gamma \widehat{V}_n + (1 - \gamma)V_n(e_n, \lambda_n)$ for~$\gamma\in(0, 1)$,
define the monotonicity projection as:
\begin{equation*}
V_{n+1}(e, \lambda)
= 
 \left\{
 \begin{array}{ll}
\widehat{U}_n, & \text{ if }e = e_n, \lambda = \lambda_n,\\
\widehat{U}_n\vee V_n(e, \lambda), & \text{ if }e = e_n, \lambda > \lambda_n,\\
\widehat{U}_n\land V_n(e, \lambda), & \text{ if }e = e_n, \lambda < \lambda_n,\\
V_n(e, \lambda), & \text{ if }e \not= e_n;
 \end{array}
\right.
\end{equation*}
\item (accuracy control): if~$\lVert V_{n+1} - V_n\lVert \leq \mathfrak{tol}$, end the scheme; 
otherwise go to Step~$2$ incrementing~$n$ to~$n+1$.
\end{enumerate}
%%%%%%%%%%%%%%%%%%%%%%%%%%%%%%%%%%%%%%%%%%%%%%
\subsection{Optimal strategy}\label{sec:OSresult}
In this section, we provide the results of the optimal decision rule computed in~\eqref{eq:OptimalSPolicy}, 
in which the value function is approximated through the numerical scheme in Section~\ref{sec:OptmStrtg}.
To begin with,
since we are dealing with the optimal liquidation problem of a child order, 
we set the size of the child order~$\chi = 2$ and the maturity~$T = 10$ 
(throughout this section, order size is measured in numbers of unit size and time is measured in seconds), 
both of which are relatively small.
Furthermore, we apply the parameters for the stylised limit order book model estimated in Section~\ref{sec:EstMd} together with the market parameters~$\rho = 1$ and~$\overline{v} = 9$ (see~\eqref{eq:p_PFun} and~\eqref{eq:mifg} for the definitions) and the tolerance level~$\mathfrak{tol}= 0.001$ in the numerical scheme.
Indeed, Proposition~\ref{prop:VFrdcfm} together with~\eqref{eq:OptimalSPolicy} indicates that the optimal decision rule depends on 
the price move direction~$j$, 
the volumes at the best prices~$v^b$ and~$v^a$, 
remaining inventory~$y$ 
and time to maturity~$\lambda$,
and is irrelevant to the ask price in tick size~$p$ and the executed limit order volume in the previous queueing race~$z$.
Moreover, the parameter estimation results in Table~\ref{tab:ParaEst0msLtc} and~\ref{tab:ParaEst1msLtc}, together with those in Figure~\ref{fig:YHOOf1msltc}, 
indicate that the agent's latency (denoted by~$\mathfrak{lat}$) also affect her optimal trading strategy.

Figure~\ref{Fig:OS0lct} shows the optimal policy as a function of~$v^b$, $v^a$, $j$ and~$\mathfrak{lat}$ by fixing~$y=2$ and~$\lambda = 10$,
where the agent's admissible trading strategies are given by~\eqref{def:actionspace} as:
\begin{equation*}
(m, l)
\in
 \left\{
 \begin{array}{ll}
\left\{(2, 0), (1, 0), (1, 1), (0 ,0), (0, 1), (0, 2)\right\}, & \text{ if }v^b > 1,\\
\left\{(0, 0), (0, 1), (0, 2)\right\}, & \text{ if }v^b = 1.
 \end{array}
\right.
\end{equation*}
Comparing the subfigures horizontally and vertically, we observe the following:
\begin{itemize}
\item 
The trading strategy that executes part of the child order, either through a limit order
$(m, l) = (0, 1)$ or through a market order $(m, l) = (1, 0)$, or does nothing $(m, l) = (0, 0)$, 
is never optimal in all scenarios. 
Generally speaking, in the situations where the best ask volume is low and the best bid volume is high
(corresponding to the top-left part of the subfigures), it it expected that the price will soon increase and the agent will choose to wait or to trade partially as her best choice.
However, since the trading horizon is quite short and the intensity rate for the incoming market orders is relatively low,
it seems that the agent would rather post limit orders in order to increase the execution probability 
than wait for better opportunities.
On top of that, this model does not consider the risk of adverse selection, 
so that posing limit orders is basically at no additional cost.
\item 
The queue imbalance of the best prices, defined as~$I := (v^b - v^a)/(v^b + v^a)$, is regarded as a powerful and effective predictor of the short-term price movements~\cite{gould2015queue, yang2016reduced} and is incorporated into the optimal market making strategy~\cite{cartea2015enhancing}.
However, we observe no clear relationship between the queue imbalance and the choice of the optimal strategy in all scenarios,
which may imply that queue imbalance should not be the only consideration in building the optimal execution strategy.
Reason for this result may come from Assumption~\ref{ass:AdmissibleTradingS}\eqref{itm:Strategythird} that the agent sticks to a `no cancellation' rule, so that the best bid and ask queue follow different dynamics.
On the contrary, the volume at the best ask price individually plays the most decisive part in the selection of the optimal strategy: 
the larger the best ask volume, the more aggressive trading strategy the agent will employ. 
In particular, when the best ask volume~$v^a \leq 6$, 
the optimal strategy is always~$(m , l) = (0, 2)$,
indicating the value of queue position for limit orders~\cite{moallemi2016model}.
Besides, volume at the best bid price also contributes to determining the optimal strategy, 
in particular when the best ask volume is high and the best bid volume is low 
(corresponding to the bottom-right part of the subfigures).
In such situations, the optimal decision rule normally chooses to take all the available liquidity through market orders in case the price soon moves against the agent's favour.
However, when the best bid volume~$v^b\geq 10$, 
the pattern of the optimal strategy is unchanged in all scenarios.
\item 
The optimal strategy is no more aggressive after a price decline than after a price increase.
This is mainly because the cancellation rate is lower at the best ask price after a price increase, 
which increases the execution risk of the agent's limit order, so the the agent prefers to use a market order in this case.
\item 
The optimal policy is no more aggressive when the agent has no latency than one-millisecond latency. 
On the one hand, 
this result comes as the cancellation rate is higher at the best ask price when there is no latency,
which increases the execution probability of the agent's limit order.
On the other hand, 
suppose the liquidation process enters into the next round of queueing race, 
in which the volumes at the best prices change dramatically within the first one millisecond,
an agent with zero latency can take most advantage of the speed to occupy a good queue position in the new queueing race.
By contrast, an agent with one-millisecond latency is less likely to get a high time priority in the new queue,
and therefore prefers to react more aggressively in order  to terminate the trade as soon as possible.
\end{itemize}

\begin{figure}[!htp]
\centering
	\begin{subfigure}
	\centering
	\includegraphics[scale=0.33]{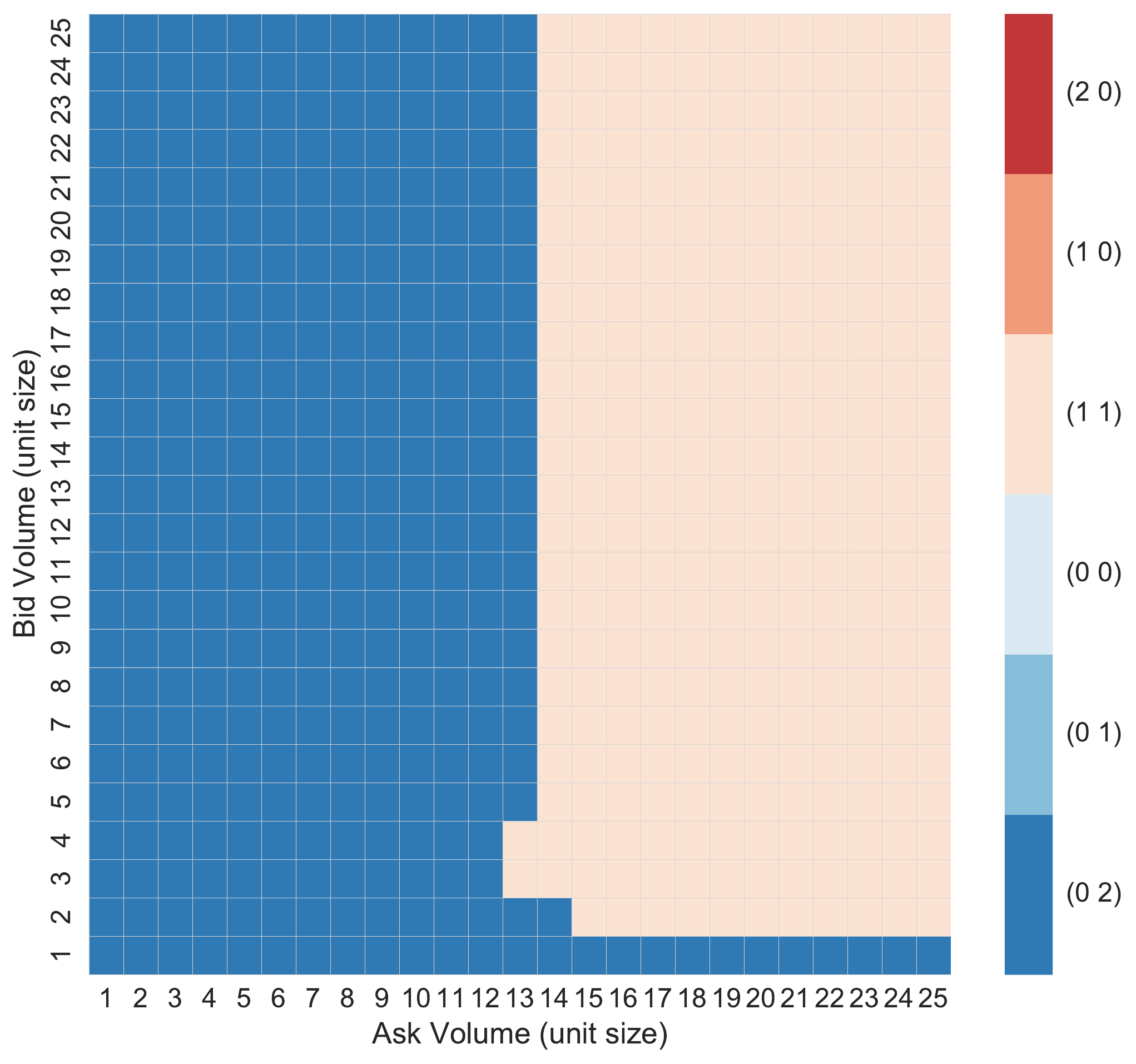}
	\end{subfigure}
	\begin{subfigure}
	\centering
	\includegraphics[scale=0.33]{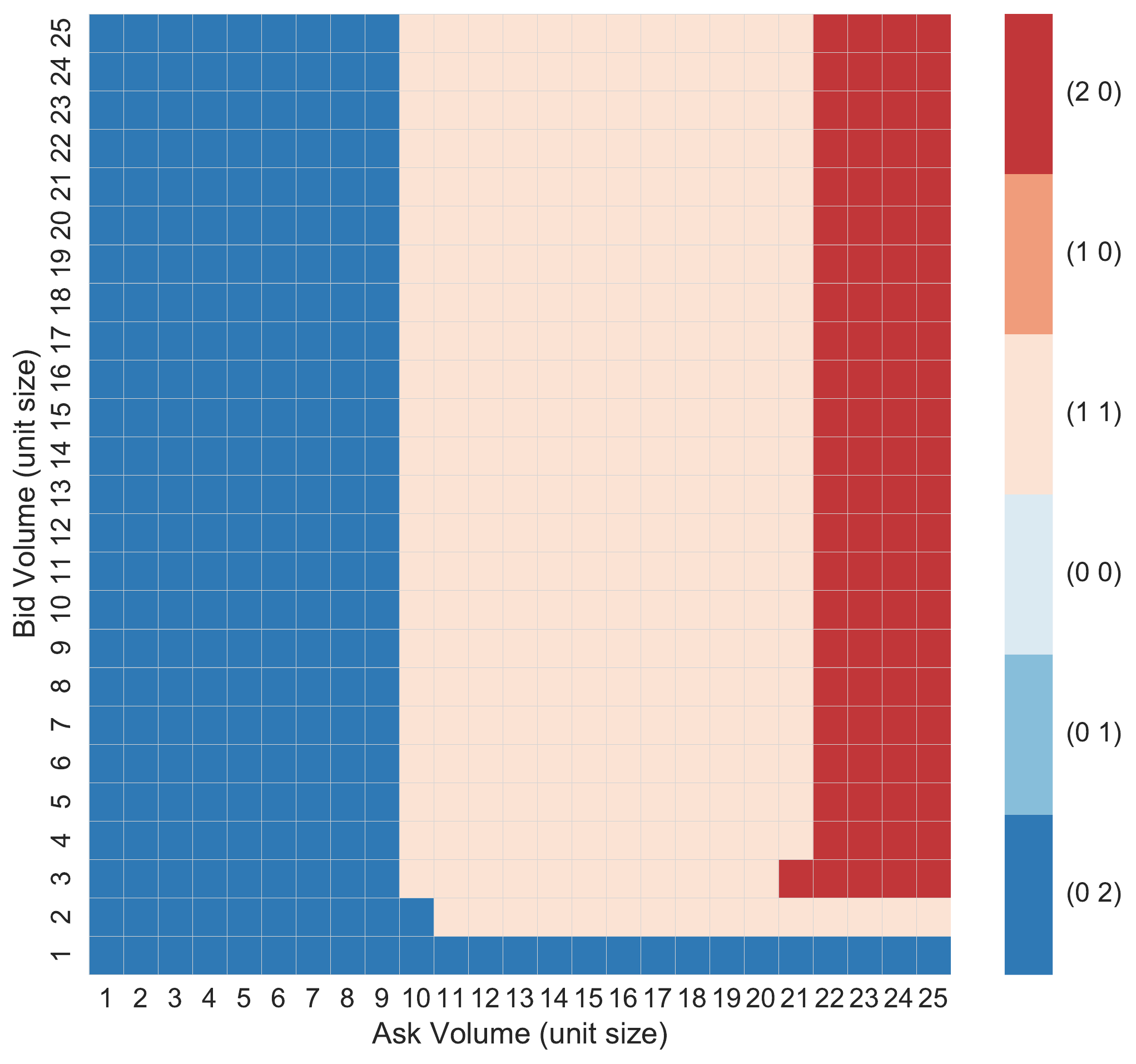}
	\end{subfigure}
	\begin{subfigure}
	\centering
	\includegraphics[scale=0.33]{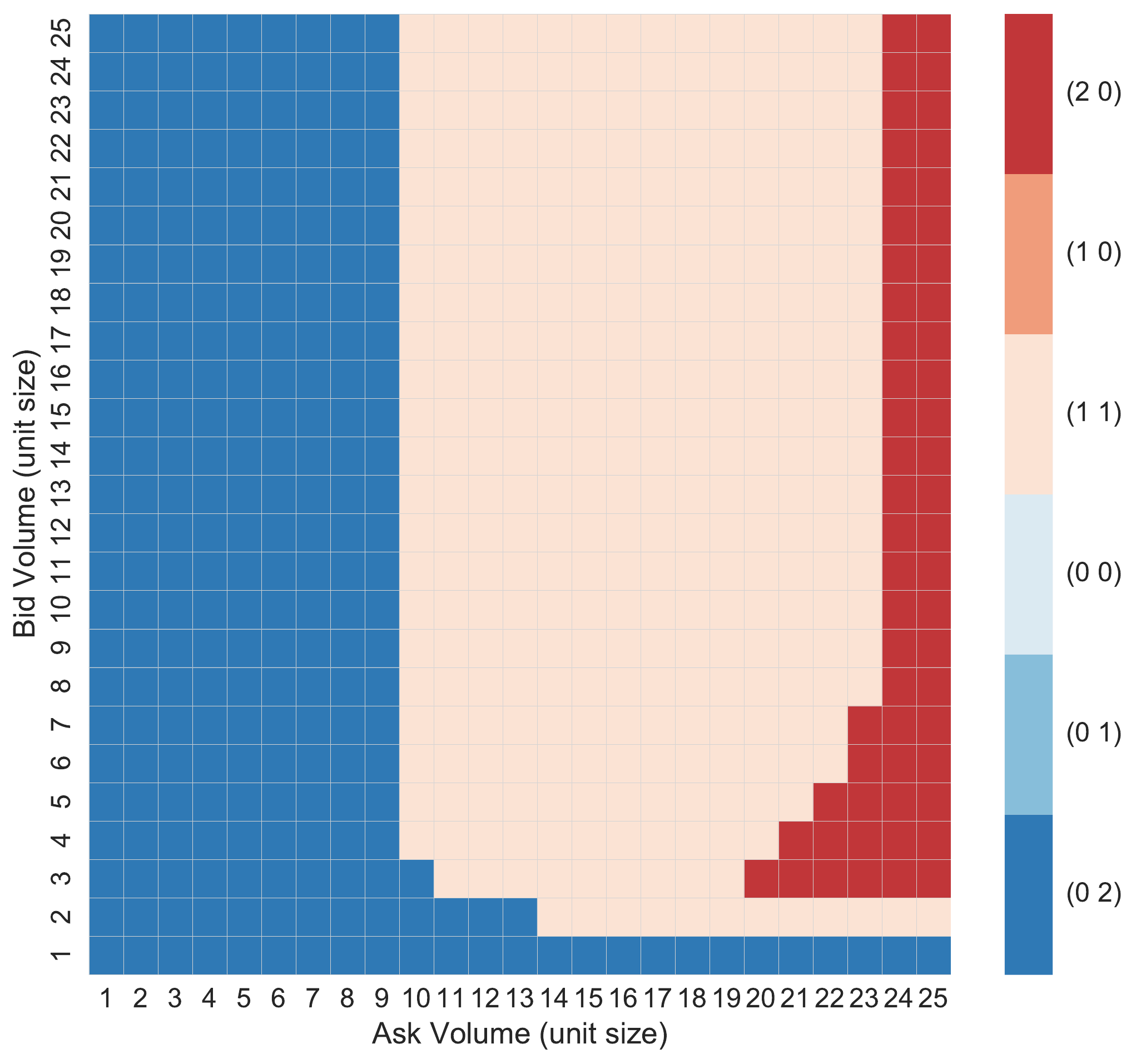}
	\end{subfigure}
	\begin{subfigure}
	\centering
	\includegraphics[scale=0.33]{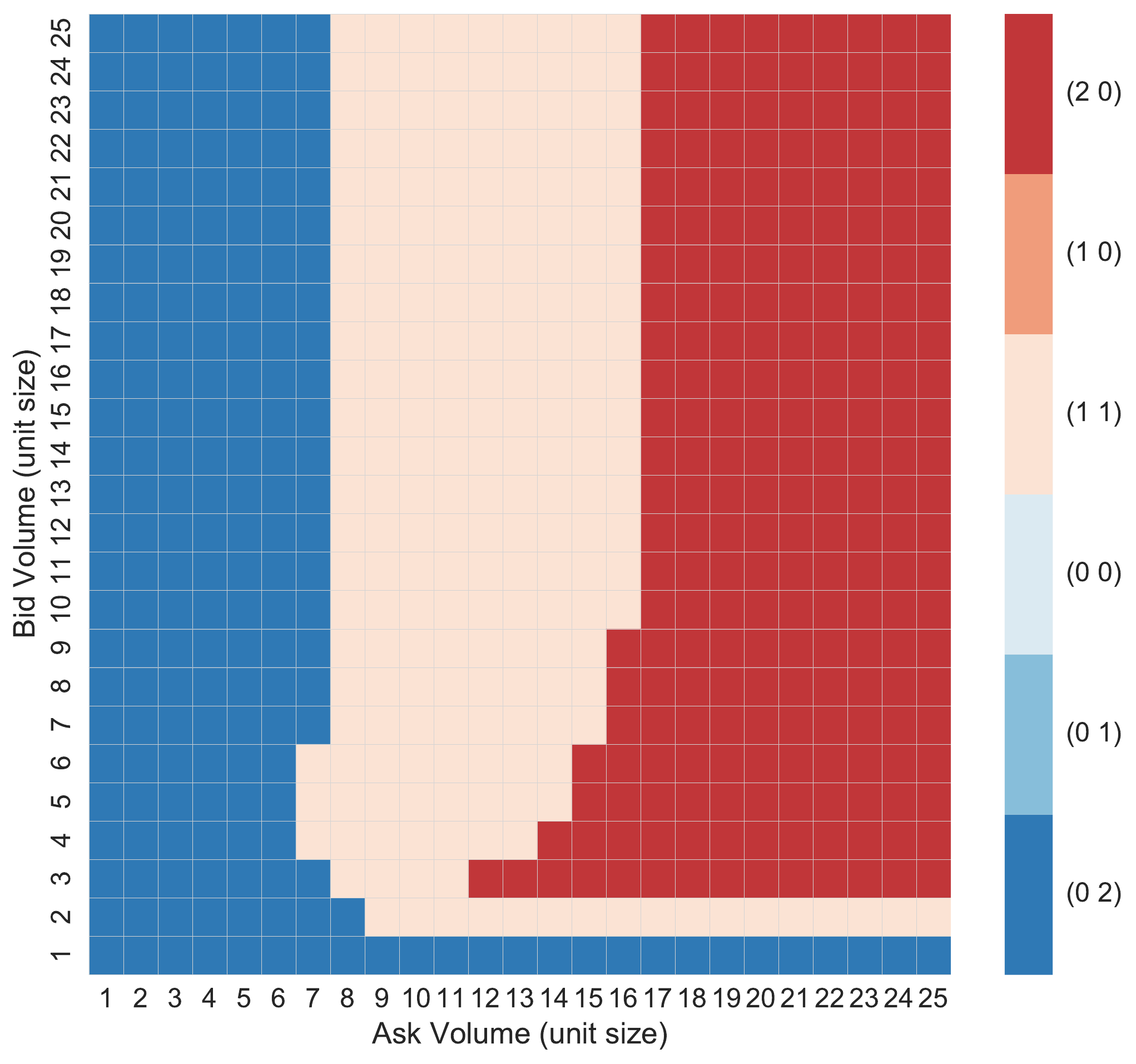}
	\end{subfigure}
\caption{Optimal policy as a function of best ask volume (x-axis: $v^a = 1, \dots, 25$),
best bid volume (y-axis: $v^b = 1, \dots, 25$), 
latency (top: $\mathfrak{lat} = 0$ms; bottom: $\mathfrak{lat} = 1$ms) 
and price move direction (left: $j = -1$; right: $j = 1$) when
fixing inventory~$y = 2$ and time to maturity~$\lambda = 10$.}
\label{Fig:OS0lct}
\end{figure}

Figure~\ref{Fig:OS1lct} shows the optimal decision rule as a function of~$v^b$, $v^a$, $j$ and~$\lambda$ (valued in~$3$ and~$10$ seconds) by fixing~$y = 1$ and~$\mathfrak{lat} = 1$ms, where the agent's admissible trading strategies are given by~\eqref{def:actionspace} as:
\begin{equation*}
(m, l)
\in
 \left\{
 \begin{array}{ll}
\left\{(1, 0),  (0 ,0), (0, 1)\right\}, & \text{ if }v^b > 1,\\
\left\{(0, 0), (0, 1)\right\}, & \text{ if }v^b = 1.
 \end{array}
\right.
\end{equation*}
In addition to the previous results, we find the agent to be more aggressive when there is less time to maturity. 

\begin{figure}[!htp]
\centering
	\begin{subfigure}
	\centering
	\includegraphics[scale=0.33]{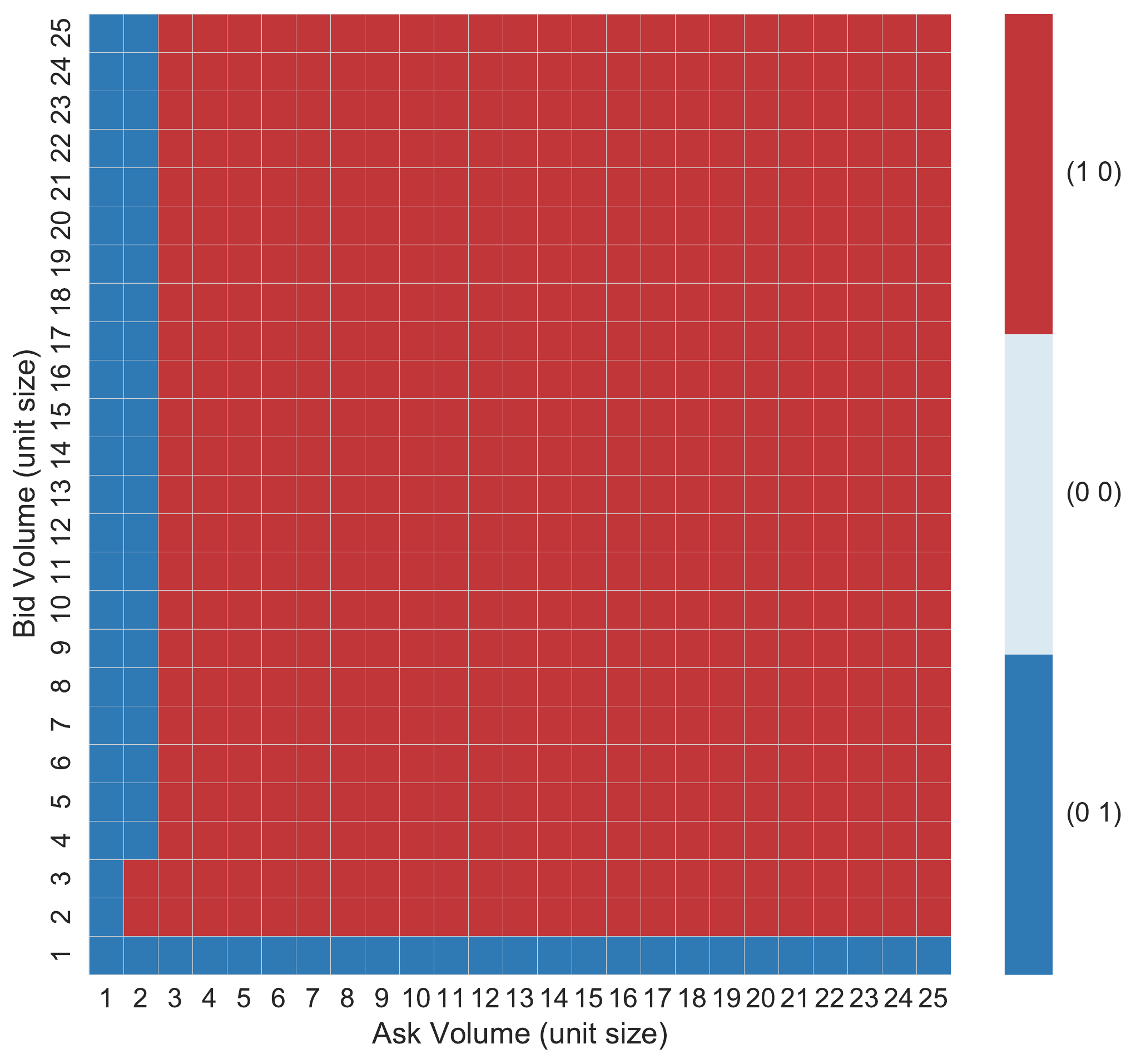}
	\end{subfigure}
	\begin{subfigure}
	\centering
	\includegraphics[scale=0.33]{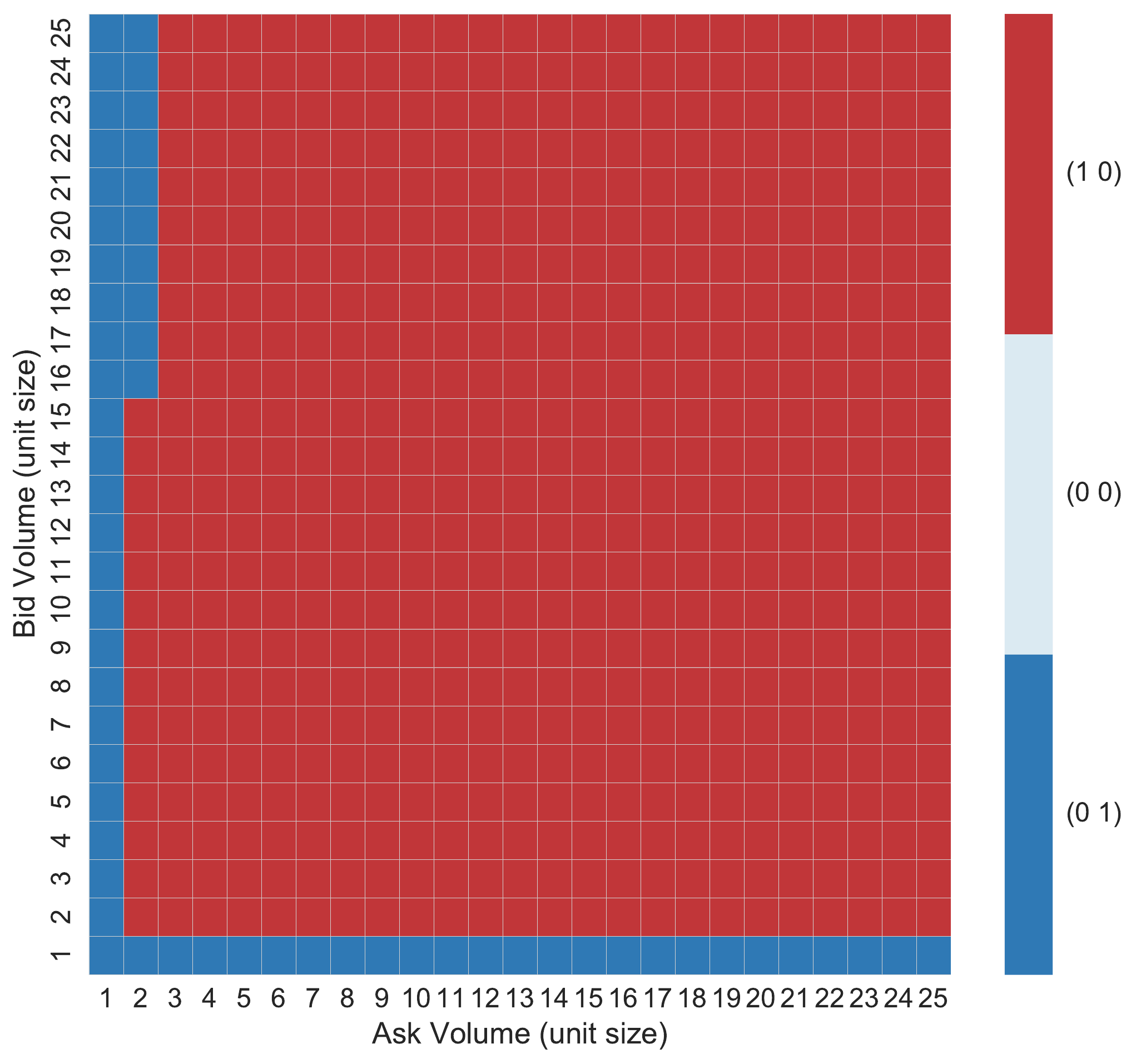}
	\end{subfigure}
	\begin{subfigure}
	\centering
	\includegraphics[scale=0.33]{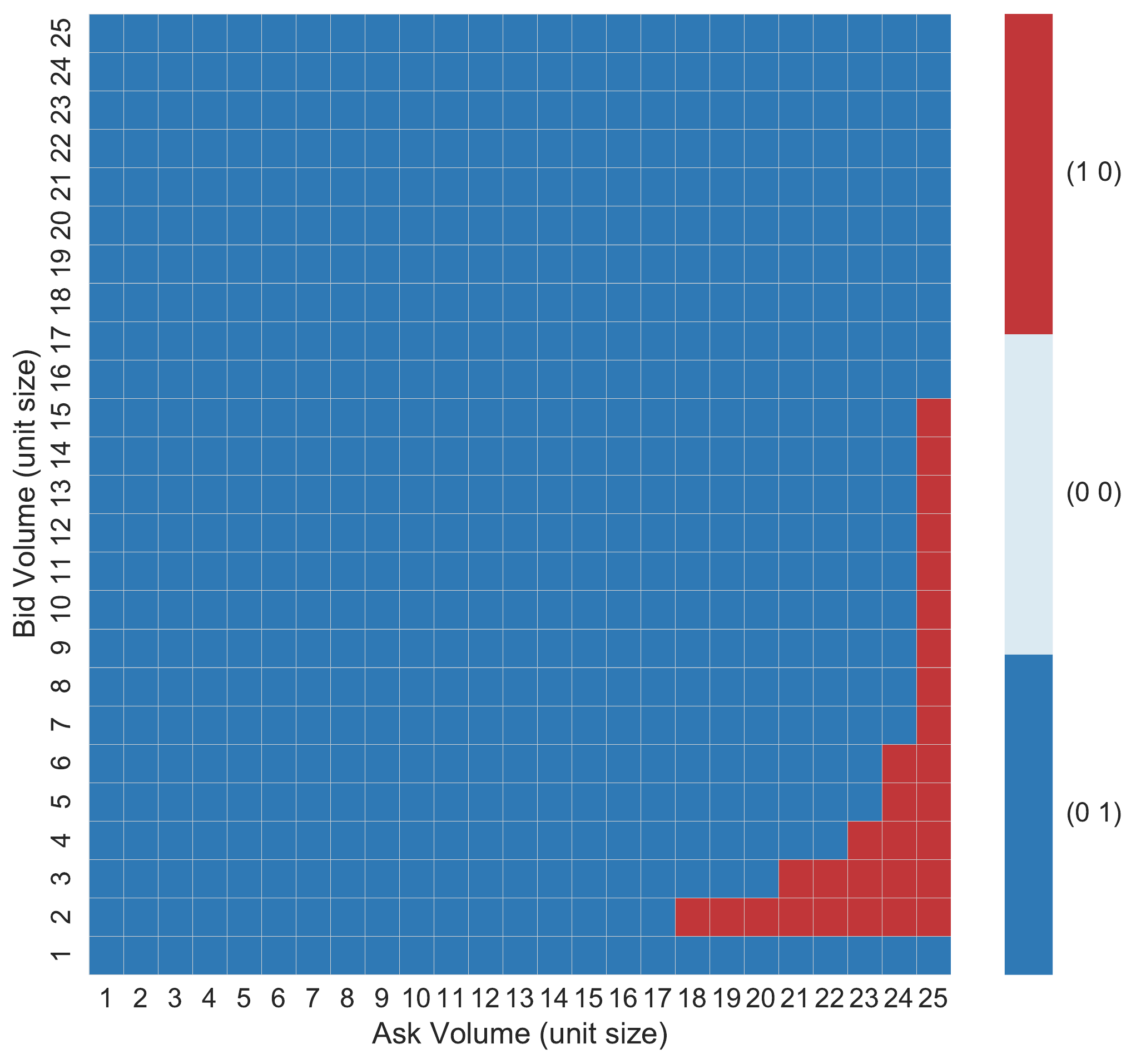}
	\end{subfigure}
	\begin{subfigure}
	\centering
	\includegraphics[scale=0.33]{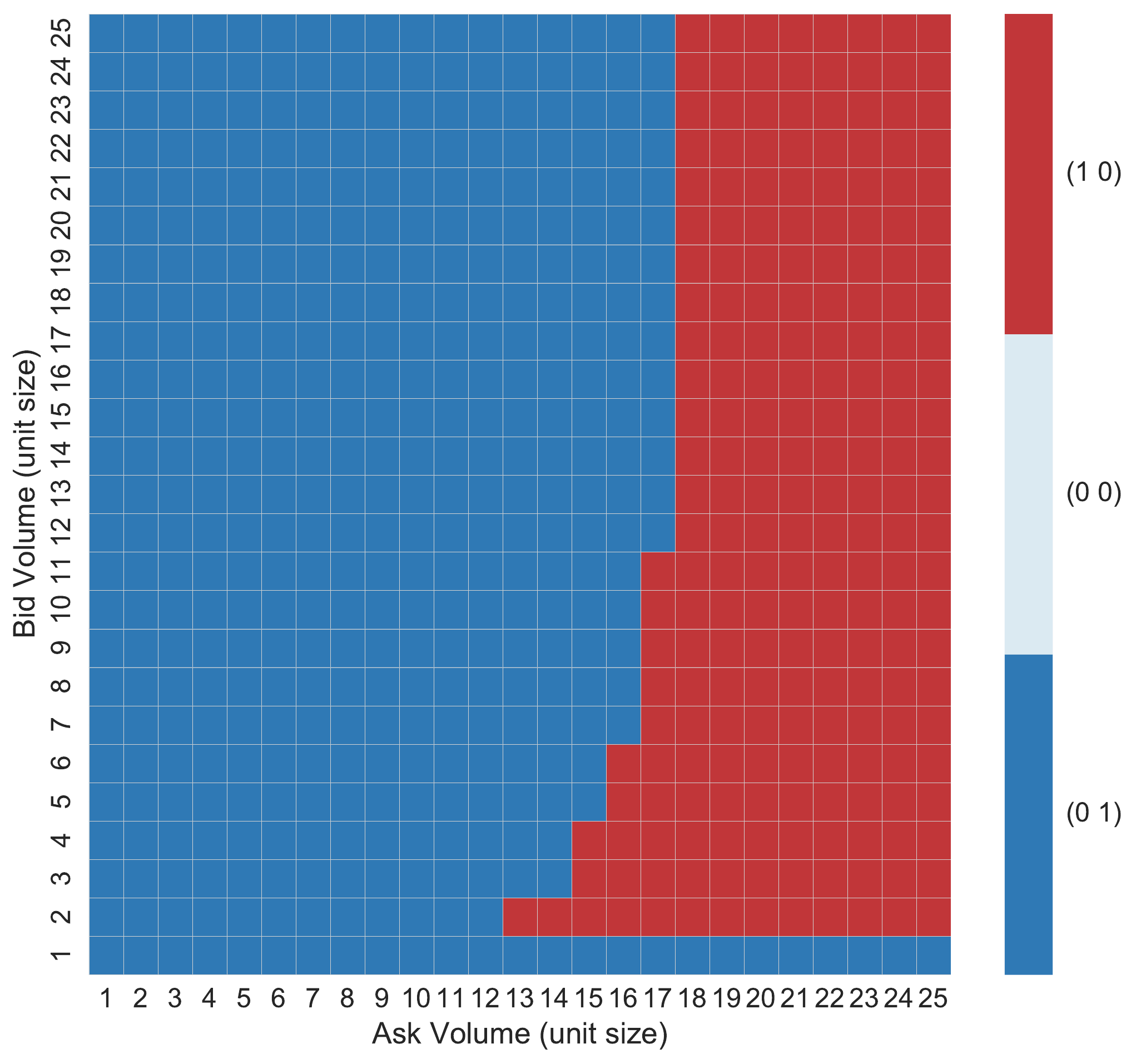}
	\end{subfigure}
\caption{Optimal policy as a function of best ask volume (x-axis: $v^a = 1, \dots, 25$),
best bid volume (y-axis: $v^b = 1, \dots, 25$), 
time to maturity (top: $\lambda = 3$; bottom: $\lambda = 10$)
and price move direction (left: $j = -1$; right: $j = 1$) when
fixing inventory~$y = 1$ and latency~$\mathfrak{lat} = 1$ms.}
\label{Fig:OS1lct}
\end{figure}

%%%%%%%%%%%%%%%%%%%%%%%%%%%%%%%%%%%%%%%%%%%%%%
\newpage
\begin{appendix}
%%%%%%%%%%%%%%%%%%%%%%%%
\section{Proof of Proposition~\ref{prop:qbarcal}}\label{app:propqbarcal}
In Scenario~[S1], the agent posts a limit order at the best ask price ($l \geq 1$),
and the best ask queue is depleted before the best bid queue ($\tilde{j} = +1$).
Hence, 
\begin{itemize}
\item the execution time of the best ask queue is less than that of the best bid queue;
\item the limit order posted by the agent must get fully executed in the queueing race;
\item the duration of the queueing race is the depletion time of the best ask queue.
\end{itemize}
Therefore, we can write
\begin{align*}
{\Qq}_{j, v, \alpha}\left(t, \tilde{j}, \tilde{z}\right)& = \PP\left(X_{n+1}\leq t, J_{n+1} = +1 \Big\lvert J_n = {j}, (V^b_n, V^a_n) = ({v}^b, {v}^a), A_n = \alpha\right)
\PP\left(Z_{n+1} = \tilde{z}\Big\lvert J_{n+1} = +1, L_n = {l}\right)\\
&=\overline{\PP}
\left(\tau_{{B}^b}>\tau_{{A}^l}, \tau_{{A}^l}\leq t\right) \ind_{\{\tilde{z} = {l}\}}
=\left\{F_{{A}^l}(t)
- \int_0^t f_{{A}^l}(u)F_{{B}^b}(u)\D u\right\}
\ind_{\{\tilde{z} = {l}\}}.
\end{align*}
%%%%%%%%%%%%%%%%%%%%%%%%%%%%%

In Scenario~[S2], the agent posts no limit order at the best ask price ($l = 0$).
The dynamics of best ask queue can be then described by the process~$B^a$,
independent of that of the best bid queue. The proof is similar to that in Scenario~[S1].
%%%%%%%%%%%%%%%%%%%%%%%%%%%%%

In Scenario~[S3], the agent posts a limit order at the best ask price ($l\geq0$),
and the best bid queue is depleted before the best ask queue ($\tilde{j} = -1$), 
while the agent's limit order gets no execution ($\tilde{z} = 0$).
Hence,
\begin{itemize}
\item the execution time of the best bid queue is less than that of one unit size of the agent's limit order together with the limit orders with higher time priority at the best ask price, and is therefore less than that of the entire best ask queue;
\item the duration of the queueing race is the depletion time of the best bid queue.
\end{itemize}
We then have
\begin{align*}
{\Qq}_{j, v, \alpha}\left(t, \tilde{j}, \tilde{z}\right)
%%%%%%%
&= {\PP}\left(X_{n+1}\leq t, Z_{n+1} = 0 \Big\lvert J_n = {j}, (V^b_n, V^a_n) = ({v}^b, {v}^a), A_n = \alpha\right)
\PP\left(J_{n+1} = -1\lvert L_n = {l}, Z_{n+1} = 0\right)\\
%%%%%%%
&= \overline{\PP}\left(\tau_{{C}^{1}}>\tau_{{B}^b},\tau_{{B}^b}\leq t\right)
= F_{{B}^b}(t) - \int_0^t f_{{B}^b}(u)F_{{C}^{1}}(u)\D u.
\end{align*}
%%%%%%%%%%%%%%%%%%%%%%%%%%%%%

In Scenario~[S4], the agent posts a limit order of one unit size at the best ask price ($l = 1$),
the best bid queue is depleted before the best ask queue ($\tilde{j} = -1$) and the agent's limit order gets executed ($\tilde{z} = 1$).
According to Remark~\ref{rem:ddotQsemiMarkov},
we have
${\Qq}_{j, v, \alpha}\left(t, \tilde{j}, \{0, 1\}\right)
=
{\Qq}_{j, v, \alpha}\left(t, \tilde{j}, 0\right)
+
{\Qq}_{j, v, \alpha}\left(t, \tilde{j}, 1\right)$, so that
\begin{align*}
{\Qq}_{j, v, \alpha}\left(t, -1, 1\right)
 & = \PP\left(X_{n+1}\leq t, J_{n+1} = -1 \Big\lvert J_n = {j}, (V^b_n, V^a_n) = ({v}^b, {v}^a), A_n = ({m}, 1)\right)-
{\Qq}_{j, v, \alpha}\left(t, -1, 0\right)\\
 & = \overline{\PP}\left(\tau_{{A}^{1}} > \tau_{{B}^b},\tau_{{B}^b}\leq t\right)
 - {\Qq}_{j, v, \alpha}\left(t, -1, 0\right)\\
&=
F_{{B}^b}(t) - \int_0^t f_{{B}^b}(u)F_{{A}^{1}}(u)\D u
 - \left[F_{{B}^b}(t) - \int_0^t f_{{B}^b}(u)F_{{C}^1}(u)\D u\right]\\
 & = \int_0^t f_{{B}^b}(u)\left[F_{{C}^{1}}(u) - F_{{A}^{1}}(u)\right]\D u.
\end{align*}
%%%%%%%%%%%%%%%%%%%%%%%%%%%%

In Scenario~[S5], the best bid queue is depleted before the best ask queue ($\tilde{j} = -1$),
while~$\tilde{z}\in\{1, \dots, {l}-1\}$ out of~${l}>1$ unit size of the agent's limit order gets executed when this queueing race terminates.
Hence,
\begin{itemize}
\item the execution time of the best bid queue lies within the interval~$[\tau_{C^{\tilde{z}}}, \tau_{C^{\tilde{z}}}+\Updelta)$,
where~$\Updelta$ is the execution time of one unit size of the agent's limit order when at the top of the queue,
which is exponentially distributed with parameter~$\mu_j^a$ and is independent of~$\tau_{C^{\tilde{z}}}$;
\item the duration of the queueing race is the depletion time of the best bid queue.
\end{itemize}
We then have
\begin{align*}
&{\Qq}_{j, v, \alpha}\left(t, \tilde{j}, \tilde{z}\right)\\
& = {\PP}\left(X_{n+1}\leq t, Z_{n+1} = \tilde{z} \Big\lvert J_n = {j}, (V^b_n, V^a_n) = ({v}^b, {v}^a), A_n = \alpha\right)
{\PP}\left(J_{n+1} = -1\lvert L_n^a = {l}, Z_{n+1} = \tilde{z}\right)\\
&=\overline{\PP}\Big(
\tau_{{C}^{\tilde{z}}} \leq \tau_{{B}^b} <\tau_{{C}^{\tilde{z}}}+\Updelta, \tau_{B^b}\leq t
\Big)
=\int_0^\infty\int_0^\infty
\overline{\PP}\left(\tau_{{B}^b}\in[u, u+\nu), \tau_{{B}^b}\leq t\right)
f_{{C}^{\tilde{z}}}(u)\overline{\PP}(\Updelta \in \D\nu)\D u\D \nu\\
& = \int_0^t \int_0^{t-u}\left[F_{{B}^b}(u+\nu) - F_{{B}^b}(u)\right]
f_{{C}^{\tilde{z}}}(u) \overline{\PP}(\Updelta \in \D\nu)\D \nu\D u
+ \int_0^t \int_{t-u}^\infty \left[F_{{B}^b}(t)-F_{{B}^b}(u)\right]
f_{{C}^{\tilde{z}}}(u)\overline{\PP}(\Updelta \in \D\nu)\D \nu\D u\\
%&=\int_0^t \int_0^{t-u}
%F_{{B}^b}(u+\nu)
%f_{{C}^{\tilde{z}}}(u)\PP(\Xi \in \D\nu)\D \nu\D u
% + F_{{B}^b}(t)\int_0^t f_{{C}^{\tilde{z}}}(u)\left[1 - \PP(\Xi \leq t - u)\right]\D u
%- \int_0^t F_{{B}^b}(u) f_{{C}^{\tilde{z}}}(u)\D u\\
& = \mu_{{j}}^a\int_0^t \E^{\mu_{{j}}^a u}f_{{C}^{\tilde{z}}}(u)
\left(\int_u^t \E^{-\mu_{{j}}^a \epsilon} F_{{B}^b}(\epsilon)\D \epsilon\right)\D u
+
\E^{-\mu_{{j}}^a t}F_{{B}^b}(t)\int_0^t \E^{\mu_{{j}}^a u}f_{{C}^{\tilde{z}}}(u)\D u
-
\int_0^t F_{{B}^b}(u) f_{{C}^{\tilde{z}}}(u)\D u\\
%& = \mu_{{j}}^a\int_0^t\!\!\int_0^\epsilon f^{*}_{{C}^{\tilde{z}}}(u)F^{*}_{{B}^b}(\epsilon)\D u\D \epsilon
%+ F^{*}_B(t)\int_0^tf^{*}_{{C}^{\tilde{z}}}(u)\D u
%- \int_0^t F^{*}_{{B}^b}(u) f^{*}_{{C}^{\tilde{z}}}(u)\D u\\
& = \mu_{{j}}^a\int_0^tF^{*}_{{B}^b}(\epsilon)\int_0^\epsilon f^{*}_{{C}^{\tilde{z}}}(u)\D u\D \epsilon
+ \int_0^t\left[F^*_{B^b}(\epsilon)\right]^\prime\int_0^\epsilon f^*_{C^{\tilde{z}}}\D u \D \epsilon\\
& = \displaystyle \int_0^t f^*_{B^b}(\epsilon)\int_0^\epsilon f^*_{C^{\tilde{z}}}(u)\D u\, \D\epsilon
\end{align*}
where~$f^{*}_{{C}^{z}}(\xi) := \E^{\mu_{{j}}^a \xi}f_{{C}^{z}}(\xi)$,
$f^{*}_{{B}^b}(\xi) := \E^{-\mu_{{j}}^a \xi}f_{{B}^b}(\xi)$
and~$F^{*}_{{B}^b}(\xi) := \E^{-\mu_{{j}}^a \xi}F_{{B}^b}(\xi)$
for any~$\xi\geq 0$ and~$z\in\NN^+$.
%%%%%%%%%%%%%%%%%%%%%%%%%%%%%%%%%%%%%

Finally, in Scenario~[S6], according to Remark~\ref{rem:ddotQsemiMarkov},
for $\left(j, v^b, v^a, \alpha\right)\in \Kk^\prime$ such that~${l}> 1$, we have
$$
\Qq_{j, v, \alpha}\left(t, -1, \{0, 1, \dots, l\}\right)
 = 
\Qq_{j, v, \alpha}\left(t, -1, 0\right)
+\displaystyle\sum_{z = 1}^{{l}-1}
\Qq_{j, v, \alpha}\left(t, -1, z\right)
+ \Qq_{j, v, \alpha}\left(t, -1, l\right),
$$
which yields the result by using~[S3] and~[S5].

%%%%%%%%%%%%%%%%%%%%%%%%%%%%%%%%%%%%%%
%%%%%%%%%%%%%%%%%%%%%%%%%%%%%%%%%%%%%%

\section{Proof of Proposition~\ref{prop:Pcal}}\label{app:Pcal}
\begin{itemize}
\item If $l=0$ and~$\zz = 0$, then
$
P(\zz\lvert (e, \alpha), \lambda) 
 = \overline{\PP}\left(\tau_{B^b}\land\tau_{B^a} > \lambda \right)
 = \overline{\PP}\left(\tau_{B^b} > \lambda\right)\overline{\PP}\left(\tau_{B^a} > \lambda\right)
 =  \overline{F}_{B^b}(\lambda)\overline{F}_{B^a}(\lambda).
$
\item If $l\geq1$, $\zz = 0$, then
$
P(\zz\lvert (e, \alpha), \lambda) 
 = \overline{\PP}\left(\tau_{B^b}\land\tau_{A^l} > \lambda, \tau_{C^1} > \lambda \right)
 = \overline{\PP}\left(\tau_{B^b} > \lambda\right)\overline{\PP}\left(\tau_{C^1} > \lambda\right)
 =  \overline{F}_{B^b}(\lambda)\overline{F}_{C^1}(\lambda).
$
\item If $l>1$ and~$\zz\in\{1, \dots, l-1\}$, then
\begin{align*}
P(\zz\lvert (e, \alpha), \lambda) & = \overline{\PP}\left(\tau_{B^b}\land\tau_{A^l} > \lambda, \tau_{C^\zz} + \Xi >\lambda\geq \tau_{C^\zz}\right)
= \overline{\PP}\left(\tau_{B^b} > \lambda\right)\overline{\PP}\left(\tau_{C^\zz} + \Xi >\lambda\geq \tau_{C^\zz}\right)\\
& =  \overline{\PP}\left(\tau_{B^b} > \lambda\right)\overline{\PP}\left[1 - \overline{\PP}\left(\tau_{C^\zz}>\lambda\right)- \overline{\PP}\left(\tau_{C^\zz} + \Xi\leq \lambda\right) \right] 
 =  \overline{F}_{B^b}(\lambda)\left[F_{C^\zz}(\lambda) - \left(F_{C^\zz} * F_{\Xi}\right)(\lambda)\right].
\end{align*}
\item If $l\geq1$ and~$\zz = l$, then
\begin{align*}
P(\zz\lvert (e, \alpha), \lambda) & = \overline{\PP}\left(\tau_{B^b}\land\tau_{A^l} > \lambda,\tau_{C^l}\leq\lambda\right) = \overline{\PP}\left(\tau_{B^b} > \lambda\right)\overline{\PP}\left(\tau_{A^l}>\lambda\geq\tau_{C^l}\right)\\
& = \overline{\PP}\left(\tau_{B^b} > \lambda\right)\left[1 - \overline{\PP}\left(\tau_{C^l}>\lambda\right) - \overline{\PP}\left(\tau_{A^l}\leq \lambda\right)\right]
 =  \overline{F}_{B^b}(\lambda)\left[F_{C^l}(\lambda) - F_{A^l}(\lambda)\right].
\end{align*}
\item According to Remark~\ref{rmk: tk}, the terminal kernel has zero value in all other scenarios.
\end{itemize}

%%%%%%%%%%%%%%%%%%%%%%%%%%%%%%%%%%%%%%
%%%%%%%%%%%%%%%%%%%%%%%%%%%%%%%%%%%%%%
\section{Proof of Proposition~\ref{prop:PropTphi}}\label{app:PropTpi}
To prove Part~\eqref{prop:PropTphifirst} of the proposition,
we can write the inequality 
$$
\lVert \Tt^\phi u\lVert 
\leq \displaystyle \sup_{(e,\lambda)\in \Ee\times\TT}
\left|r(e, \phi(e, \lambda))\right| + 
\sup_{\substack{(e,\lambda)\in \Ee\times\TT \\ \zz\in\{0, \dots, \bar{l}\}}}
\left| w(e, \phi(e, \lambda), \zz)\right|
 + \sup_{(e,\lambda)\in \Ee\times\TT}\left\lvert\sum_{\tilde{e}\in \Ee}\int_0^\lambda u(\tilde{e}, \lambda-t)Q\big(\D t, \tilde{e}\lvert (e, \phi(e, \lambda))\big)\right\lvert,
$$
for any~$\phi\in\Phi$ and~$u\in \Uu$.
The first two terms are bounded since the state space~$\Ee$ and the action space~$\Aa$ are finite.
Regarding the last term,
applying Lemma~\ref{lem:fntde} yields
\begin{align*}
\sup_{(e,\lambda)\in \Ee\times\TT}\left\lvert\sum_{\tilde{e}\in \Ee}\int_0^\lambda u(\tilde{e}, \lambda-t)Q\big(\D t, \tilde{e}\lvert (e, \phi(e, \lambda))\big)\right\lvert
 & \leq \| u\| \displaystyle \sup_{(e,\lambda)\in \Ee\times\TT}
\left\lvert \sum_{\tilde{e}\in \Ee}\int_0^\lambda Q\big(\D t, \tilde{e}\lvert (e, \phi(e, \lambda))\big)\right\lvert\\
& = \lVert u\lVert \displaystyle \sup_{(e,\lambda)\in \Ee\times\TT}
Q\big(\lambda, \Ee\lvert (e, \phi(e, \lambda))\big)\\
 & \leq \| u \|\sup_{\lambda\in\TT}(1-\E^{-2{\iota}\lambda}) = \| u \| (1-\E^{-2{\iota}T})<\infty.
\end{align*}
Therefore the codomain of~$\Tt^\phi$ is~$\Uu$.
The contraction property follows directly from~\eqref{ineq:regu1}, since 
$\|\Tt^\phi u - \Tt^\phi v \|\leq
(1-\E^{-2{\iota}T})\| u - v \|$
holds for all $u, v\in \Uu$,
and the monotonicity follows from the properties of the semi-Markov kernel.
To prove Part~\eqref{prop:PropTphisecond} of this proposition, we can write, 
for any~$(e, \lambda)\in \Ee\times\TT$ and~$\pi := \{\phi_0, \phi_1, \phi_2, \dots\}\in\Pi$,
\begin{align*}
V^{\pi}(e, \lambda) & = \sum_{n = 0}^\infty \EE^{\pi}_{(e, \lambda)}\left[
r(\Ee_n, \Aa_n)\ind_{\{\Lambda_n\geq 0\}}
 + w(\Ee_n, \Aa_n, \ZZ)\ind_{\{0\leq\Lambda_n < X_{n+1}\}}\right]\\
 & = \EE^{\pi}_{(e, \lambda)}\left[
r(\Ee_0, \Aa_0)\ind_{\{\Lambda_0\geq 0\}}
 + w(\Ee_0, \Aa_0, \ZZ)\ind_{\{0\leq\Lambda_0 < X_{1}\}}
 \right]\\
 &\qquad\qquad+\displaystyle\sum_{n = 1}^\infty\EE^{\pi}_{(e, \lambda)}\left[
r(\Ee_n, \Aa_n)\ind_{\{\Lambda_n\geq 0\}}
 + w(\Ee_n, \Aa_n, \ZZ)\ind_{\{0\leq\Lambda_n < X_{n+1}\}}\right]\\
 & = \EE^{\pi}_{(e, \lambda)}\left[ 
 \EE^{\pi}_{(e, \lambda)}\left[
 r(\Ee_0, \Aa_0)\ind_{\{\Lambda_0\geq 0\}}
 +  \sum_{\zz = 0}^\infty w(\Ee_0, \Aa_0, \zz)\ind_{\{0\leq\Lambda_0 < X_{1}, \ZZ = \zz\}}
 \Big\lvert H_0
 \right]
 \right]\\
 &\qquad\qquad+\displaystyle\sum_{n = 1}^\infty\EE^{\pi}_{(e, \lambda)}\left[
 \EE^{\pi}_{(e, \lambda)}\left[
r(\Ee_n, \Aa_n)\ind_{\{\Lambda_n\geq 0\}}
 + w(\Ee_n, \Aa_n, \ZZ)\ind_{\{0\leq\Lambda_n < X_{n+1}\}}
 \Big\lvert
 H_1, 
 \right]
 \right]\\
 & = r(e, \phi_0(e, \lambda))
 + \sum_{\zz = 0}^\infty w(e, \phi_0(e, \lambda), \zz)P(\zz\lvert (e, \phi_0(e, \lambda)), \lambda)
\\
 & \qquad + 
 \sum_{n =1}^\infty \EE^{\pi}_{(e, \lambda)}\left[
 \EE_{(\Ee_1, \Lambda_1)}^{\pi_-}\left[
 r(\Ee_{n-1}, \Aa_{n-1})\ind_{\{\Lambda_{n-1}\geq 0\}}
 + w(\Ee_{n-1}, \Aa_{n-1}, \ZZ)\ind_{\{0\leq\Lambda_{n-1} < X_{n}\}}
 \right]
 \right]\\
 & = r(e, \phi_0(e, \lambda)) + \sum_{\zz = 0}^\infty w(e, \phi_0(e, \lambda), \zz)P(\zz\lvert (e, \phi_0(e, \lambda)), \lambda)
 + \EE_{(e, \lambda)}^\pi\left[V^{\pi_-}(\Ee_1, \Lambda_1)\right]\big)
 \\
 & = r(e, \phi_0(e, \lambda)) + \sum_{\zz = 0}^\infty w(e, \phi_0(e, \lambda), \zz)P(\zz\lvert (e, \phi_0(e, \lambda)), \lambda)
 + \sum_{\tilde{e}\in \Ee}\int_0^\lambda V^{\pi_-}(\tilde{e}, \lambda-t)Q\big(\D t, \tilde{e}\lvert (e, \phi_0(e, \lambda))\big),
\end{align*}
according to~Remark~\ref{rmk:vfchange} and Theorem~\ref{thm:Tulcea}, which concludes the proof.

%\section{Proof of Proposition~\ref{prop:stationPlc}}\label{app:stationPlc}

%%%%%%%%%%%%%%%%%%%%%%%%%%%%%%%%%%%%%%%%%%%%%%
%%%%%%%%%%%%%%%%%%%%%%%%%%%%%%%%%%%%%%%%%%%%%%
\section{Proof of Proposition~\ref{prop:VFrdcfm}}\label{app:VFrdcfm}
Let ${\bf i} := (0, 0, 0, 1, 0, 0)$ and ${\bf k} := (0, 0, 0, 0, 1, 0)$, so that
$e = \overline{e} + \Delta_p {\bf i} + \Delta_z {\bf k}$,
where~$\Delta_p = p- \overline{p}$ and~$\Delta_z := z - \overline{z}$.
Define further $\hat{e} := \overline{e} + \Delta_{z}{\bf k}$. 
According to Proposition~\ref{prop:PropTphi} and~\cite[Theorem 3]{denardo1967contraction}, 
we can write
\begin{equation}\label{eq:rdc1}
V^*(\hat{e}, \lambda) = \Aa V^*(\hat{e}, \lambda)
=
V^*(\overline{e}, \lambda) + \rho\Delta_z(\overline{p} - j).
\end{equation}
With the auxiliary function
$u(\ee, \mathfrak{\lambda}) := V^*(\ee - \Delta_p {\bf i}, \lambda) + \rho\Delta_p(y+z)$
for $(\ee, \lambda)\in \Ee\times \TT$, 
simple calculations yield $\Aa u(\ee, \lambda) = u(\ee, \lambda)$  for any~$(\ee, \lambda)\in \Ee\times \TT$,
and Theorem~\ref{thm:ValueFunction} implies that
$V^*(e, \lambda) = V^*(\hat{e}, \lambda) + \rho\Delta_p (y+z)$.
Combining this with~\eqref{eq:rdc1} concludes the proof.

%%%%%%%%%%%%%%%%%%%%%%%%%%%%%%%%%%%%%%
\section{Maximum Likelihood Estimation for the Poisson Parameters}\label{App:MLE}
Fix~$\ssf\in\{a, b\}, j\in\{+1, -1\}$ and 
denote the Poisson parameters~$\mu_j^\ssf, \kappa_j^\ssf, \theta_j^\ssf$
by~$\mu, \kappa, \theta$ respectively.
Introduce the auxiliary parameters~$\mu^\prime := \mu S^l/S^m$
and~$\theta^\prime := \theta S^l/ S^c$.
Suppose we observe~$l_i$ times of limit order arrivals, 
$m_i$ times of market order arrivals and~$c_i$ times of cancellations on the~$\ssf$ side in the~$i$-th queueing race,
whose starting time is~$\tau_i$,
duration is~$d_i$
and
the volume in unit size at~$\ssf$ price at time~$t$ is~$\text{Vol}_i(t)$,
for~$i \in \{1, \dots, \# \Qf_j\}$.
The likelihood functions are then constructed as:
\begin{align*}
\mathcal{L}\left(\mu^\prime: m_1, \dots, m_{\# \Qf_j}, d_1, \dots, d_{\# \Qf_j}\right)
&:=\prod_{i =1}^{\# \Qf_j}\frac{{(\mu^\prime d_i)}^{m_i}}{m_i !}\E^{-\mu^\prime d_i},\\
\mathcal{L}\left(\kappa: l_1, \dots, l_{\# \Qf_j}, d_1, \dots, d_{\# \Qf_j}\right)
&:=\prod_{i =1}^{\# \Qf_j}\frac{{(\kappa d_i)}^{l_i}}{l_i !}\E^{-\kappa d_i},\\
\mathcal{L}\left(\theta^\prime: c_1, \dots, c_{\# \Qf_j}, \Theta(d_1), \dots, \Theta(d_{\# \Qf_j})\right)
&:=\prod_{i =1}^{\# \Qf_j}\frac{{\Theta(d_i)}^{c_i}}{c_i !}\E^{-\Theta(d_i)},
\end{align*}
where $\Theta(d_i):=\theta^\prime \int_{\tau_i}^{\tau_i+d_i}\text{Vol}_i(t)\D t$.
Taking logarithms, and cancelling the derivatives yield the optima~\eqref{eq:MaxLikelihoodResult}
with
$$
N_{s,j}^{\varpi} = \sum_{i=1}^{\# \Qf_j}\varpi_i,\quad\text{for }\varpi\in \{m,l,c\}\qquad
D_{s,j} = \sum_{i=1}^{\# \Qf_j}d_i,\qquad
V_{s,j} = \sum_{i=1}^{\# \Qf_j}\int_{\tau_i}^{\tau_i+d_i}\text{Vol}_i(t)\D t.
$$
\end{appendix}

%%%%%%%%%%%%%%%%%%%%%%%%%%%%%%%%%%%%%%

%%%%%%%%%%%%%%%%%%%%%%%%%%%%%%%%%%%%%%%%%%%%%
\end{document}